\documentclass[11pt]{article}
\usepackage[utf8]{inputenc}

\usepackage{amsmath, amssymb, amsthm, amsfonts}
\usepackage{mathtools}
\usepackage{thm-restate}

\usepackage[dvipsnames,svgnames,x11names]{xcolor}
\definecolor{ForestGreen}{rgb}{0.1333,0.5451,0.1333}
\usepackage[colorlinks=true,linkcolor=ForestGreen,citecolor=blue]{hyperref}

\usepackage[margin=1in]{geometry}
\usepackage{graphics}
\usepackage{pifont}
\usepackage{tikz}
\usepackage{bbm}
\usepackage[T1]{fontenc}

\usetikzlibrary{arrows.meta}
\usepackage{environ}
\usepackage{framed}
\usepackage{url}
\usepackage{algorithm}
\usepackage[algo2e,linesnumbered,lined,ruled,boxed]{algorithm2e}
\usepackage{algpseudocode}
\usepackage[noabbrev,capitalize,nameinlink]{cleveref}
\crefname{equation}{}{}
\usepackage[labelfont=bf]{caption}
\usepackage{cite}
\usepackage{framed}
\usepackage[framemethod=tikz]{mdframed}
\usepackage{appendix}
\usepackage{graphicx}
\usepackage[textsize=tiny]{todonotes}
\usepackage{tcolorbox}
\allowdisplaybreaks[1]

\newcommand\remove[1]{}

\newtheorem{lemma}{Lemma}[section]
\newtheorem{theorem}{Theorem}
\newtheorem*{lemma*}{Lemma}
\newtheorem{corollary}[lemma]{Corollary}
\newtheorem*{corollary*}{Corollary}

\newtheorem{assumption}[lemma]{Assumption}

\theoremstyle{definition}

\newtheorem*{theorem*}{Theorem}
\newtheorem{definition}[lemma]{Definition}

\newtheorem*{rem*}{Remark}

\newcommand{\eps}{\varepsilon}
\newcommand{\R}{\mathbb{R}}
\newcommand{\ls}{\lesssim}

\newcommand{\E}{\mathbb{E}}

\newcommand{\mA}{\mathbf{A}}
\newcommand{\mD}{\mathbf{D}}
\crefname{algocf}{Algorithm}{Algorithms}
\renewcommand{\l}{\langle}
\renewcommand{\r}{\rangle}
\newcommand{\supp}{\mathsf{supp}}
\newcommand{\Z}{\mathbb{Z}}
\newcommand{\D}{\mathbb{D}}

\newcommand{\gs}{\gtrsim}
\newcommand{\bbC}{\mathbb{C}}
\newcommand{\Image}{\mathsf{Image}}

\newcommand{\A}{\Sigma}
\newcommand{\B}{\Gamma}
\newcommand{\C}{\Phi}
\renewcommand{\a}{\sigma}
\renewcommand{\b}{\gamma}
\renewcommand{\c}{\phi}
\newcommand{\master}{\mathfrak{m}}
\newcommand{\ama}{\a_{\master}}
\newcommand{\bma}{\b_{\master}}
\newcommand{\cma}{\c_{\master}}
\renewcommand{\aa}{\kappa}
\newcommand{\bb}{\lambda}
\newcommand{\cc}{\mu}

\renewcommand{\bar}{\overline}
\renewcommand{\hat}{\widehat}

\newcommand{\cG}{\mathcal{G}}
\newcommand{\val}{\mathsf{val}}
\newcommand{\wt}{\widetilde}
\newcommand{\rep}{\mathsf{rep}}
\newcommand{\cA}{\mathcal{A}}

\renewcommand{\bar}{\overline}
\newcommand{\diag}{\mathsf{diag}}
\renewcommand{\deg}{\mathsf{deg}}
\newcommand{\effdeg}{\mathsf{effdeg}}
\newcommand{\eff}{\mathsf{eff}}
\newcommand{\const}{\mathsf{const}}
\newcommand{\cP}{\mathcal{P}}
\newcommand{\cQ}{\mathcal{Q}}
\newcommand{\mV}{\mathbf{V}}
\renewcommand{\Re}{\mathsf{Re}}
\newcommand{\modest}{\mathsf{modest}}
\newcommand{\moddeg}{\mathsf{moddeg}}
\newcommand{\TT}{\mathcal{T}}
\newcommand\skipi{{\vskip 10pt}}

\begin{document}

\title{Parallel Repetition for $3$-Player $\mathsf{XOR}$ Games}

\author{Amey Bhangale\thanks{Department of Computer Science and Engineering, University of California, Riverside. Supported by the Hellman Fellowship award.}
	\and
    Mark Braverman\footnote{Department of Computer Science at Princeton University. Supported by the National Science Foundation under
the Alan T. Waterman Award, Grant No. 1933331.}
    \and
	Subhash Khot\thanks{Department of Computer Science, Courant Institute of Mathematical Sciences, New York University. Supported by
		the NSF Award CCF-1422159, NSF CCF award 2130816, and the Simons Investigator Award.}
	\and
    Yang P. Liu\thanks{School of Mathematics, Institute for Advanced Study, Princeton, NJ. Partially supported by NSF DMS-1926686}
    \and 
	Dor Minzer\thanks{Department of Mathematics, Massachusetts Institute of Technology. Supported by NSF CCF award 2227876 and NSF CAREER award 2239160.}}

\date{\vspace{-3ex}}
\clearpage\maketitle

\begin{abstract}
In a $3$-$\mathsf{XOR}$ game $\mathcal{G}$, the verifier samples a challenge $(x,y,z)\sim \mu$
where $\mu$ is a probability distribution over $\Sigma\times\Gamma\times\Phi$, 
and a map $t\colon \Sigma\times\Gamma\times\Phi\to\mathcal{A}$ for a finite
Abelian group $\mathcal{A}$ defining a constraint. 
The verifier
sends the questions $x$, $y$ and 
$z$ to the players Alice, Bob and 
Charlie respectively, receives answers $f(x)$, $g(y)$ and $h(z)$ 
that are elements in $\mathcal{A}$ 
and accepts if $f(x)+g(y)+h(z) = t(x,y,z)$. The value, $\mathsf{val}(\mathcal{G})$, of the game
is defined to be the maximum 
probability the verifier accepts 
over all players' strategies.

We show that if $\mathcal{G}$ is a $3$-$\mathsf{XOR}$
game with value strictly less than $1$, whose underlying distribution over
questions $\mu$ does not admit Abelian
embeddings into $(\mathbb{Z},+)$, then
the value of the $n$-fold repetition
of $\mathcal{G}$ is exponentially decaying. That is, there exists 
$c=c(\mathcal{G})>0$ such that $\mathsf{val}(\mathcal{G}^{\otimes n})\leq 2^{-cn}$. 
This extends a previous result of 
[Braverman-Khot-Minzer, FOCS 2023] showing 
exponential decay for the GHZ game.
Our proof combines tools from additive
combinatorics and tools from discrete
Fourier analysis.
\end{abstract}

\newpage
\setcounter{tocdepth}{2}
\tableofcontents

\normalsize
\pagebreak
\pagenumbering{arabic}

\section{Introduction}
\label{sec:intro}
In a $k$-player game $\cG$, a verifier chooses a question $x := (x^{1}, x^{2}, \dots, x^{k}) \sim \mu$, for some distribution $\mu$ over $\A_1 \times \A_2 \times \dots \times \A_k$, for finite alphabets $\A_1, \dots, \A_k$. The verifier sends $x^i$ to player $i$, who responds with an answer $f^i(x^i)$. The verifier accepts based on the evaluation of some known predicate $V(x, \{f^1(x^1), \dots, f^k(x^k)\})$. We define the \emph{value} $\val(\cG)$ of the game $\cG$ as the maximum probability of winning such a game for the players.

A natural question that arises is how the value of a game decays under parallel repetitions. The \emph{$n$-fold parallel repetition} of $\cG$, which we denote as $\cG^{\otimes n}$, is the game where the verifier samples $n$ questions $(x_1, \dots, x_n) \sim \mu^{\otimes n}$, and sends the $i$-th coordinate of all of $(x_1, \dots, x_n)$ to player $i$ simultaneously. The players win the repeated game if they win each one of the $n$ instances.

Parallel repetition of $2$-player games is well-understood by now.
Originally suggested by~\cite{FRS94} 
as a way to amplify the advantage in interactive 
protocols, it is now known that the
value of $2$-player games decays
exponentially if the value of the
game is strictly smaller than $1$.
This result was first proved by Raz~\cite{Raz98} using information theoretic techniques. Subsequent 
works simplified the proofs and improved it quantitatively~\cite{H09,Rao11,DS14,BG15}. We remark that most of these
works, with the exception of~\cite{DS14}, follow the information
theoretic proof of Raz. The work~\cite{DS14} suggests an 
analytical approach to the problem 
in the case the game $\mathcal{G}$ 
is a projection game.\footnote{We remark that in the majority of applications, and especially in 
the area of probabilistically checkable proofs, one is often 
only concerned with projection games.}
While it is natural to suspect that the (na\"{i}ve) bound  $\val(\cG^{\otimes n}) \le \val(\cG)^n$ should hold, it turns out to be false \cite{Fei91,FV02}.  Even this failure is well understood, and it is known to be related to the
geometry of tilings of high dimensional Euclidean space 
(see~\cite{FKO07,KORW08,AK09,BMtilesym}), and more specifically to the 
fact that there are bodies of volume 
$1$ and surface area $\Theta(\sqrt{n})$ that tile the 
Euclidean space $\mathbb{R}^n$.
Parallel repetition of $2$-player games has many applications in the theory of interactive proofs \cite{BGKW88}, communication complexity \cite{PRW97,BBCR13,BRWY13}, quantum information~\cite{CHTW04,BBLV13}, and hardness of approximation~\cite{FGLSS96,ABSS97,ALMSS98,AS98,BGS98,Fei98,Has01,Khot02a,Khot02b,GHS02,DGKR05,DRS05}.

The situation for $k$-player games, 
where $k\geq 3$, is much more poorly understood. Even for $k = 3$, the best bound known general bound for $\val(\cG^{\otimes n})$ is inverse Ackermann, which follows from quantitative versions of the Density Hales-Jewett theorem \cite{V96,DHJ12}. Since proving effective bounds on the value of parallel repetition of $k$-player games seems like a very challenging task, researchers tried to focus their attention on more limited classes of games. In~\cite{DHVY17}, the authors consider
the class of \emph{connected games}, 
and prove that an exponential decay 
holds for them. Their proof uses information-theoretic techniques similar to those used in the context of $2$-player version. Here, a game is connected if the following graph on $\supp(\mu)$ is connected: connect $x, x' \in \supp(\mu)$ with an edge if $x, x'$ differ in a single coordinate. More quantitatively, one
considers a weighted version of 
this graph, where the weights correspond to re-sampling a coordinate
in the distribution $\mu$; with this in mind, connectedness is concerned with 
the second eigenvalue of this graph
being bounded away from $1$. The
authors of~\cite{DHVY17} then suggested
that to make more progress on multi-player parallel repetition one must
tackle games that are not connected, 
and suggested the GHZ game as an example of such game.

The GHZ game is a special $3$-player
game that has its origin in physics (where it is known as quantum pseudo-telepathy). In the GHZ game the
verifier samples a challenge 
$(x,y,z)\in \mathbb{F}_2^3$ such 
that $x+y+z = 0$, sends $x$, $y$ 
and $z$ to Alice, Bob and Charlie 
respectively, receives answers 
$f(x),g(y),h(z)\in\mathbb{F}_2$, 
and accepts if $f(x) + g(y) + h(z) = x\lor y\lor z$. The GHZ game is easily
seen to have classical value $3/4$, 
but all existing techniques seemed to
fail to prove an effective upper 
bounds on the value of 
$\textsf{GHZ}^{\otimes n}$. Holmgren 
and Raz~\cite{HR20} were the first
to tackle this problem, and they
showed that 
$\val(\textsf{GHZ}^{\otimes n})
\leq n^{-\Omega(1)}$ using a combination of information theory
and Fourier analysis. Later works
simplified the proof for the $\textsf{GHZ}$ game~\cite{GHMRZ21} 
and used these techniques to prove
polynomial bounds for all $3$-player
games with binary questions~\cite{GHMRZ22,GMRZ22}, namely games where $\supp(\mu)\subseteq \{0,1\}^3$.
Finally, in~\cite{BKM22} it was
shown that $\val(\textsf{GHZ}^{\otimes n})\leq 2^{-\Omega(n)}$ using techniques
from additive combinatorics. 
We defer the reader to \cref{subsec:barriers} for further discussion.

\skipi
The primary goal of this work is
to study parallel repetition for the very natural class of 3-$\mathsf{XOR}$ which includes within it the $\textsf{GHZ}$ game. 
We show that the value of a 3-$\textsf{XOR}$ games decays 
exponentially provided that the underlying distribution of question
does not admit any Abelian embeddings
into $(\mathbb{Z},+)$. Our proof uses
new techniques from Fourier analysis~\cite{BKMcsp1,BKM2,BKMcsp3,BKM4} as well as powerful results from 
additive combinatorics~\cite{Gowers}.

\subsection{$3$-\textsf{XOR} Games, Embeddings and Connectedness}
\label{subsec:embedding}
To state our main result formally we require some setup. Suppose that  $\A, \B, \C$ are finite alphabets, and let $\mu$ be a distribution on $\A \times \B \times \C$. Let $t: \supp(\mu) \to \cA$ be a target function, where $\cA$ is a finite Abelian group. The pair $(\mu,t)$ defines a 
3-$\textsf{XOR}$ game $\cG$, wherein 
the verifier samples a challenge $(x,y,z)\sim \mu$, sends $x$, $y$ 
and $z$ to Alice, Bob and Charlie 
respectively, receives answers $f(x),g(y),h(z)\in\cA$ and 
accepts if $f(x)+g(y)+h(z) = t(x,y,z)$.

The \emph{value} of an $n$-fold repeated game under this setup is defined as the maximum over all $f:\A^n \to \cA^n, g: \B^n\to \cA^n, h: \C^n \to \cA^n$ of
\begin{align} \Pr_{(x,y,z)\sim\mu^{\otimes n}}\left[f(x)_i + g(y)_i + h(z)_i = t(x_i,y_i,z_i) \enspace \forall \enspace i = 1, 2, \dots, n \right]. \label{eq:gamevalue} \end{align}
Let $\cG^{\otimes n}$ be the $n$-fold repeated game, with value $\val(\cG^{\otimes n})$. We are not able to establish parallel repetition for \emph{all} $3$-player $\mathsf{XOR}$ games, and we need to make some assumptions on the structure of the distribution $\mu$. We discuss why removing these assumptions may prove to be quite challenging in \cref{subsec:discussion}.

\begin{definition}[Embeddings]
\label{def:embedding}
An embedding of a distribution $\mu$ on $\A \times \B \times \C$ into an abelian group $A$ consists of three maps $\a: \A \to A, \b: \B \to A, \c: \C \to A$ such that $\a(x) + \b(y) + \c(z) = 0$ for all $(x,y,z) \in \supp(\mu)$.
\end{definition}
We say that $\mu$ has no $\Z$-embeddings if all embeddings of $\mu$ into $\Z$ are \emph{trivial} (in the sense that $\a, \b, \c$ must be constant maps). Throughout, we will assume that $\mu$ has no non-trivial $\Z$-embeddings.
If $\mu$ has no $\Z$-embeddings, then it also satisfies a mild connectivity property.
\begin{definition}[Pairwise-connected]
\label{def:pairwise}
We say that a distribution $\mu$ on $\A \times \B \times \C$ is $(x,y)$-\emph{pairwise-connected} if the graph on $\A \times \B$ with an edge between $(x,y)$ if there is some $z$ with $(x,y,z) \in \supp(\mu)$ is connected. We say that $\mu$ is pairwise-connected if it is $(x,y)$, $(y,z)$, and $(x,z)$ pairwise-connected.
\end{definition}
\begin{lemma}
\label{lemma:pairwise}
If $\mu$ has no $\Z$-embeddings, then $\mu$ is pairwise-connected.
\end{lemma}
\begin{proof}
If $\mu$ is not $(x,y)$-connected, then there are partitions $\A \to \A' \cup \A''$ and $\B = \B' \cup \B''$ such that if $(x,y,z) \in \supp(\mu)$, then $(x,y) \in \A' \times \B'$ or $(x,y) \in \A'' \times \B''$. Thus, we can define the embedding $\a(x) = \b(y) = 1$ for $x \in \A'$, $y \in \B''$, $\a(x) = \b(y) = -1$ otherwise, and $\c(z) = 0$ for all $z$.
\end{proof}
We note that this notion of pairwise-connectivity is much weaker than connectivity in the context of \cite{DHVY17}, or even player-wise connectivity as defined in \cite{GHMRZ22}. In particularly, the GHZ game is pairwise-connected but neither connected or player-wise connected. This explains why the techniques for showing parallel repetition for the GHZ game differ significantly from the analyses of (player-wise) connected games, which are information-theoretic in nature.

\subsection{Our Main Result}
The main result of work is that games $\cG$ with $\mathsf{XOR}$ predicates over input distributions $\mu$ without $\Z$-embeddings, have exponential decay under parallel repetition.
\begin{theorem}
\label{thm:main}
For a $3$-player game $\cG$ over distribution $\mu$ with $\mathsf{XOR}$ predicates, if $\val(\cG) < 1$ and $\mu$ has no $\Z$-embeddings, then there is a constant $c := c(\cG,\mu) > 0$ such that $\val(\cG^{\otimes n}) \le 2^{-cn}$.
\end{theorem}
Beyond the result itself, our proof method shows a connection between parallel repetition, at least in the case of $\mathsf{XOR}$ games, and techniques/questions in additive combinatorics. Specifically, our proof leverages an analytical stability result from~\cite{BKM4} 
for $3$-wise correlations over pairwise-connected distributions with no $\Z$-embeddings. Thus, we (very speculatively) suspect that better understanding of these stability results (or extensions to higher arities $k > 3$) may lead to more general $3$-player or $k$-player parallel repetition results. 
We defer a more detailed discussion
of our techniques to~\cref{sec:overview}.

\subsection{Discussion of \texorpdfstring{$\Z$}{Z}-Embeddings and Pairwise-Connectivity}
\label{subsec:discussion}

Showing stronger stability results for $3$-wise correlations would lead to improvements for open problems in additive combinatorics.
This presents a sort of barrier for removing the no $\Z$-embedding assumption from $\mu$.
The reason we need to assume that $\mu$ has no $\Z$-embedding is because our proof relies on the stability result in \cite{BKM4}, 
which is currently known to hold 
only for distributions $\mu$ with no $\Z$-embedding.
We will see the reason more precisely in \cref{subsec:master1}. We remark
though that at the current state of
affairs, the assumption that $\mu$ 
has no $\mathbb{Z}$-embeddings is 
more to be seen as a technical barrier; for all we know, analogous
stability results should hold even if
$\mu$ is just assumed to be pairwise
connected, and it stands to reason
such results would translate to effective parallel repetition 
theorems using ideas form the current
paper.

We suspect that removing the pairwise-connectedness assumption on $\mu$ 
will be much more challenging, though. In this case, it is not
even clear what the stability result for non pairwise-connected distributions $\mu$ should even be. Moreover, proving such a stability result
with effective parameters will 
lead to the first effective bounds 
for the density Hales-Jewett theorem
for combinatorial lines of length $3$,\footnote{To see this, one takes the distribution $\mu$ to be supported on $\{(0,0,0),(1,1,1),(2,2,2),(0,1,2)\}$. The fact that a subset $A\subseteq \{0,1,2\}^n$ contains
no combinatorial lines can then be phrased as $\E_{(x,y,z)\sim \mu^{\otimes n}}[1_A(x)1_A(y)1_A(z)] = o(1)$, and an effective inverse 
theorem would then give a structural result about $A$ which 
may yield to a density increment argument as in~\cite{BKM5}.} 
resolving a major open problem in 
extremal combinatorics.

\section{Overview of Proof of \texorpdfstring{\cref{thm:main}}{thmmain}}
\label{sec:overview}
In this section we give an overview
of the proof of \cref{thm:main}.
\subsection{High-Level Structure of Argument}
\label{overview:highlevel}
By the fundamental theorem of finite Abelian groups, to show \cref{thm:main} it suffices to consider the case where $\cA$ is a cyclic group of prime power order, say $\cA = \Z_{m}$ where $m=p^k$ for a prime $p$. For intuition throughout this section, the reader can think of the game as being over $\Z_2$, though all our proofs work over other cyclic groups. Even if the original game is over $\Z_2$, our proofs ultimately require defining functions over larger cyclic groups (e.g., $\Z_{2^k}$), just as the analysis of the GHZ game of \cite{BKM22} looked at functions over $\Z_4$.

\subsubsection{Step 1: Reducing to an Analytic Statement and Nonembeddable Targets}
The first step of the proof is to arithmetize the probability that a set of strategies $f, g, h$ win $\cG^{\otimes n}$ as an analytic quantity. Throughout, let $\omega$ be the complex $m$th root of unity.
\begin{lemma}\label{lemma:identity}
Let $f: \A^n \to \Z_m^n, g: \B^n \to \Z_m^n, h: \C^n \to \Z_m^n$ be strategies. For a vector $S \in \Z_m^n$ define $F_S(x) = \omega^{\sum_{i\in [n]} S_i f(x)_i}$, and $G_S, H_S$ analogously. Define $T_S(x,y,z) = \omega^{-\sum_{i\in[n]} S_i t(x_i, y_i, z_i)}$. The winning probability of the strategies $f, g, h$ is exactly given by:
\begin{align*} \E_{S \sim \Z_m^n} \E_{(x,y,z)\sim\mu^{\otimes n}} \left[ F_S(x)G_S(y)H_S(z)T_S(x,y,z) \right].
\end{align*}
\end{lemma}
\cref{lemma:identity} follows from simply expanding out the definitions of $F_S, G_S, H_S, T_S$ and reversing the order of the expectations.

From here, it is beneficial to consider $2$-player $\mathsf{XOR}$ games, and give a simple and direct proof of parallel repetition. For the following result, let $\D$ denote the unit disk in the complex plane, and $T(x, y) = \omega^{-t(x,y)}$.
\begin{lemma}[Two-dimensional case]
\label{lemma:2dcase}
For a \emph{connected} two player $\mathsf{XOR}$ game $\cG$ with distribution $\mu$ over $\A \times \B$ and target function $t: \supp(\mu) \to \Z_m$ with value less than $1$, there is constant $c := c(\cG,\mu) > 0$ such that for any $F: \A^n \to \D$, $G: \B^n \to \D$ we have
\[ \left|\E_{(x,y) \sim \mu^{\otimes n}}\left[F(x)G(y) \prod_{i=1}^n T(x_i,y_i) \right] \right| \le 2^{-cn}. \]
\end{lemma}
The intuition is that the quantity we wish to bound is an inner product of $F, G$ with respect to $T^{\otimes n}$ (where $T$ is thought of as a matrix), and we know that singular value of matrices tensorize.
\begin{proof}
Let $\mD_x = \diag(\mu_x)$ and $\mD_y = \diag(\mu_y)$. Note that $\|(\mD_x^{1/2})^{\otimes n} F\|_2, \|(\mD_y^{1/2})^{\otimes n} G\|_2 \le 1$. Thus, it suffices to show that for all $A: \A \to \D$ and $B: \B \to \D$ with $\|\mD_x^{1/2}A\|_2 = \|\mD_y^{1/2}B\|_2 = 1$, that
\[ \left| \E_{(x,y) \sim \mu}[A(x)B(y)T(x,y)]\right| < 1. \]
This implies then that the maximum singular value of $\mD_x^{-1/2}T \mD_y^{-1/2}$ is less than $1$, and the lemma follows because the maximum singular value tensorizes.
Note that by the Cauchy-Schwarz inequality,
\[ \left| \E_{(x,y) \sim \mu}[A(x)B(y)T(x,y)]\right| \le \|\mD_x^{1/2}A\|_2 \|\mD_y^{1/2}B\|_2 = 1, \]
with equality only if $|A(x)| = |B(y)|$ and $A(x)B(y)T(x,y) = 1$ for all $(x, y) \in \supp(\mu)$. By connectivity of $\mu$, the former condition and $\|\mD_x^{1/2}A\|_2 = 1$ implies that $|A(x)| = 1$ for all $x \in \A$, and $|B(y)| = 1$ for all $y \in \B$. Let $A(x) = e^{2\pi i a(x)}, B(y) = e^{2\pi i b(y)}$. Then the latter condition gives that $a(x) + b(y) = t(x,y)/m \pmod{1}$ for all $(x, y) \in \supp(\mu)$. By shifting $a(x), b(y)$ to be multiples of $1/m$, this implies that $t$ has value $1$, a contradiction.
\end{proof}
This implies parallel repetition for $2$-player $\mathsf{XOR}$ games. Indeed, we can assume the game is connected. Using \cref{lemma:2dcase}, for any $S \in \Z_m^n$ we have $|\E_{(x,y)\sim\mu^{\otimes n}}[F_S(x) G_S(y) T_S(x,y)]| \le 2^{-c|S|}$, and as $|S| \ge \Omega(n)$ on average over the choice $S\sim \mathbb{Z}_m^n$ we get from  \cref{lemma:identity} an exponential decay for the value of the $2$-player $\textsf{XOR}$ game.

\skipi

Can we carry out such an argument for $3$-player games? The first piece we would need is an analogue of \cref{lemma:2dcase}. In fact, the na\"{i}ve generalization of \cref{lemma:2dcase} to $3$-functions is \emph{false}, as we discuss in more detail in \cref{subsec:barriers}. The condition we actually require on $t$ and $T$ comes from looking at the end of the proof of \cref{lemma:2dcase}, where we argued that if $A(x)B(y) = T(x,y)$ for all $(x,y) \in \supp(\mu)$, then $t$ has a strategy of value $1$. This is false in the $3$-player case. Instead, for the $3$-player setting, the assumption we require on $T$ is exactly that there are no functions $A: \A \to \D$, $B: \B \to \D$, $C: \C \to \D$ such that $A(x)B(y)C(z)T(x,y,z) = 1$ for all $(x,y,z) \in \supp(\mu)$. In this case, we say that $T$ is \emph{nonembeddable over $\D$}. 
In the case that $T$ is nonembeddable over $\D$, we are able to show that this statement tensorizes by following the work of \cite{BKM2} (though some adaptations are needed to get a proper exponential decay):
\begin{restatable}{theorem}{nocomplex}
    \label{thm:nocomplex}
    Let $T:\supp(\mu) \to \D$ be pairwise-connected and nonembeddable over $\D$. Then there is a constant $c := c(T,\mu) > 0$ such that for all functions $F:\A^n \to \D, G: \B^n \to \D, H: \C^n \to \D$, 
    \[ \left|\E_{(x,y,z)\sim \mu^{\otimes n}}\left[F(x)G(y)H(z) \prod_{i=1}^n T(x_i,y_i,z_i) \right] \right| \le 2^{-cn}. \]
\end{restatable}
The proof of \cref{thm:nocomplex} is substantially more challenging than its two-dimensional counterpart stated in~\cref{lemma:2dcase}, and we overview the ideas in \cref{overview:analytic}.
Combining this with \cref{lemma:identity} shows that for target functions $t$ that are nonembeddable over $\D$, parallel repetition with exponential decay holds.

\subsubsection{Step 2: Barriers to Analytic Argument, and Linear Structure}
\label{subsec:barriers}
Next, we discuss a case where $T$ is embeddable over $\D$, but the game $\cG$ still has value less than $1$. In this case, we are not able to 
apply \cref{thm:nocomplex}, and
we must resort to a different set of techniques.  

The most natural example of a 3-$\textsf{XOR}$ game in this case is the $\textsf{GHZ}$ game, defined as follows. The distribution $\mu$ over $\Z_2^3$ is uniform over $\{(x,y,z) : x+y+z=0 \},$ and strategies $f, g, h$ win on $x, y, z$ if $f(x) + g(y) + h(z) = x \vee y \vee z \pmod{2}.$ It is easy to check that the value of this game is $3/4$. The tensor $T$ defined by $T(x, y, z) = (-1)^{x \vee y \vee z}$, and has values \[ T(0, 0, 0) = 1, T(1, 1, 0) = -1, T(1, 0, 1) = -1, T(0, 1, 1) = -1. \] This is embeddable over $\D$ by taking $A = B = C$, and $A(x) = i^x = e^{\frac12\pi i \cdot x}$. Another way of viewing this is that
\[ \frac12(x+y+z) = x \vee y \vee z = t(x,y,z) \pmod{2} \enspace \forall \enspace (x, y, z) \in \supp(\mu). \]
Thus, $2t(x, y, z) = x+y+z \pmod{4}$. Thus, the winning condition $f(x) + g(y) + h(z) = t(x, y, z) \pmod{2}$ can be rewritten $(2f(x) - x) + (2g(y) - y) + (2h(z) - z) = 0 \pmod{4}$. This is the key observation in~\cite{BKM22}, 
and combining it with powerful tools from additive combinatorics (which we overview in \cref{overview:additive}) yields exponential decay for  parallel repetition of the $\textsf{GHZ}$ game. 

Our argument in the general case 
that $T$ is embeddable follows a
similar strategy and uses additive
combinatorics as well. Note that if $A(x)B(y)C(z)T(x,y,z) = 1$ for \[ A(x) = e^{2\pi i a'(x)}, B(y) = e^{2\pi i b'(y)}, C(z) = e^{2\pi i c'(z)}, T(x,y,z) = e^{-2\pi i \cdot t(x,y,z) / m}, \] 
then we get that $a'(x) + b'(y) + c'(z) = t(x,y,z)/m \pmod{1}$. Assuming that $a(x), b(y), c(z)$ are all rational with denominator dividing $Nm < \infty$ (which we can assume) we get that (writing $a'(x) = a(x)/(Nm)$ and similar for $b'(y), c'(z)$),
\[ a(x) + b(y) + c(z) = N \cdot t(x, y, z) \pmod{Nm}, \] for functions $a: \A \to \Z_{Nm}, b: \B \to \Z_{Nm}, c: \C \to \Z_{Nm}$. Thus the winning condition $f(x) + g(y) + h(z) = t(x, y, z) \pmod{m}$ can be rewritten as:
\[ (N \cdot f(x) - a(x)) + (N \cdot g(y) - b(y)) + (N \cdot h(z) - c(z)) = 0 \pmod{Nm}. \]
This is a very similar setup to the $\textsf{GHZ}$ game we described above, with one critical difference. In the $\textsf{GHZ}$ game, the distribution $\mu$ itself was a subspace of $\Z_2^3$ consisting of points with sum $0$, while our input distribution $\mu$ is still somewhat arbitrary. This arbitrariness prevents us from applying tools 
from additive combinatorics, and 
to remedy the situation, 
we must gain more control over the 
structure of the distribution $\mu$. In the next section, we discuss how to use ideas from \cite{BKM4} to reduce to the case where $\mu$ is uniform over $\{(x,y,z) \in \cA^3 : x+y+z = 0 \}$, for some (possibly very large) finite Abelian group $\cA$.

\subsubsection{Step 3: The Master Embedding and Reducing to Functions Constant on the Master Embedding}
\label{subsec:master1}
The group $\cA$ discussed at the end of the previous section will be called the \emph{master group}, which we define now. Fix a sufficient large integer $r$, and let $(\a_1,\b_1,\c_1), \dots, (\a_s,\b_s,\c_s)$ be all embeddings of $\mu$ into cyclic groups $\cA_1, \dots, \cA_s$ of (prime power) order at most $r$, with $\a_i(x^*) = \b_i(y^*) = \c_i(z^*) = 0$ for some fixed $x^* \in \A, y^* \in \B, z^* \in \C$, which we can assume by shifting each embedding.
Then, we define the master embedding as $\ama(x) := (\a_1(x), \dots, \a_s(x))$ (\cref{def:master}), and  set $\cA \subseteq \cA_1 \times \dots \times \cA_s$ as the subgroup generated by $\{\ama(x) \colon x \in \A \}$. The key property of $\ama$ and $\cA$ is that if $\mu$ has no $\Z$-embeddings, then for $r$ sufficiently large the following holds. For all $x, x' \in \A$, $\a(x) = \a(x')$ for all embeddings $\a$ if and only if $\ama(x) = \ama(x')$ (\cref{lemma:same}), and analogously for $\bma, \cma$.

Recall that at the end of \cref{subsec:barriers}, we rewrote the winning strategy of the game as $\tilde{f}(x) + \tilde{g}(y) + \tilde{h}(z) = 0$ for $(x, y, z) \in \supp(\mu)$, where $\tilde{f}(x) = N \cdot f(x) - a(x)$, $\tilde{g}(y) = N \cdot g(y) - b(y)$, and $\tilde{h}(z) = N \cdot h(z) - c(z)$. Thus, $f, g, h$ win the game exactly when $\tilde{f}, \tilde{g}, \tilde{h}$ act like an embedding on $\mu$. 
The general philosophy of the master
embedding asserts that to get such 
behaviour, the values of the functions $\tilde{f}(x)$, 
$\tilde{g}(y)$ and $\tilde{h}(z)$ must 
morally depend only on $\ama(x)$, 
$\bma(y)$ and $\cma(z)$ respectively.
Indeed, we are able to transform the functions $\tilde{f}$, $\tilde{g}$ and $\tilde{h}$ to functions with
similar winning probabilities 
that are constant on the master embeddings. Abusing of notations 
we assume that $\tilde{f}$, $\tilde{g}$ and $\tilde{h}$ are 
such functions to begin with, namely $\tilde{f}(x) = \tilde{f}(x')$ if $\ama(x) = \ama(x')$, and analogously for $\tilde{g}$ and $\tilde{h}$. This way, regardless of the original distribution $\mu$, we can imagine that we are working with the
set of challenges $\{(\ama(x),\bma(y),\cma(z))~|~(x,y,z)\in \supp(\mu)\}$, and by 
changing variables we may assume that
to begin with we have that
\[ \supp(\mu) \subseteq \{(x, y, z) \in \cA \times \cA \times \cA \colon x+y+z = 0\}, \] 
(where $x+y+z = 0$ because the master embeddings satisfies this equation). Formalizing this intuition is much more challenging, and is the main content of \cref{sec:complex}. In particular, it requires a version of our analytic result \cref{thm:nocomplex} in this setting (\cref{thm:highmodest}). We remark that a version of this was shown in \cite{BKM4}, but in order to achieve our tighter quantitative bounds of $\exp(-cn)$ for parallel repetition, we need to give a tighter analysis and a slightly different proof adapted to our setting.

In the above paragraph we gave intuition on how to reduce to the case where $\supp(\mu)$ is a subset of $(x, y, z) \in \cA \times \cA \times \cA$ with $x+y+z = 0$. However, we actually wish to have $\supp(\mu)$ to be exactly this set. To achieve this, we use path tricks (discussed below in \cref{overview:analytic}) to enlarge the input alphabets and distribution $\mu$, without changing the master group $\cA$, and without reducing the winning probability of the game by much. This idea of using path tricks to enlarge the distribution was exactly used in \cite{BKM4} in the setting of showing analytic theorems, and an identical approach works for parallel repetition too.

\subsubsection{Step 4: Applying an Additive Combinatorics Argument}
\label{overview:additive}
Finally, we use techniques from additive combinatorics to solve the case where $\supp(\mu)$ is exactly $\{(x,y,z) \in \cA \times \cA \times \cA \colon x+y+z = 0\}.$ Our main theorem is the following.
\begin{restatable}{theorem}{additive}
\label{thm:additive}
Let $\cA$ be a finite abelian group and $t: \cA^3 \to \Z_{p^k}$. Let $N > 1$ be minimal so that there are functions $a, b, c: \cA \to \Z_{p^k N}$ with $a(x) + b(y) + c(z) = N \cdot t(x,y,z) \pmod{p^k N}$ for all $x + y + z = 0$. Then there is a constant $c = c(t,\cA) > 0$ such that for any $f, g, h: \cA^n \to \Z_{p^k}$,
\[ \Pr_{\substack{(x,y,z) \in (\cA^n)^3 \\ x+y+z=0}}[f(x)_i + g(y)_i + h(z)_i = t(x_i,y_i,z_i) \pmod{p^k}, i = 1, 2, \dots, n] \le 2^{-cn}. \]
\end{restatable}
A simpler version was given for the GHZ game in \cite{BKM22}. At a high level (using the notation from the previous sections), the argument goes as follows. If $\tilde{f}(x) + \tilde{g}(y) + \tilde{h}(z) = 0$ with probability at least $\eps$ over $x+y+z = 0$, $(x, y, z) \in \cA \times \cA \times \cA$, then using tools from additive combinatorics such as the Balog-Szemeredi-Gowers lemma, we deduce significant structural facts about $\tilde{f}, \tilde{g}, \tilde{h} \colon \cA \to \Z_{Nm}$ (\cref{lemma:freiman}). We ultimately combine this with the fact that $\tilde{f}$ is not completely arbitrary, as it must take the form $\tilde{f}(x) = N \cdot f(x) - a(x)$, to arrive at a contradiction.

\subsection{Proof of Main Analytic Results}
\label{overview:analytic}

\subsubsection{Changing the distribution via random restrictions}
In this section overview the tools needed to show our analytic results, focusing on \cref{thm:nocomplex} 
(our other analytical result,~\cref{thm:highmodest} involves similar ideas).
A very useful trick throughout our proofs will be that we can alter the distribution over questions $\mu$ (without changing its support), at the cost of reducing the number of repetitions $n$ by a constant factor. For example, we can change the distribution to uniform over its support. Formally, we can show:
\begin{lemma}
\label{lemma:rr}
Let $\cG_{\mu}$ be a game with question distribution $\mu = \delta \mu' + (1-\delta)\nu$ for distributions $\mu', \nu$.
Let $\cG_{\mu'}$ be the game with the same questions over distribution $\mu'$. Then \[ \val(\cG_{\mu}^{\otimes n}) \le \val(\cG_{\mu'}^{\otimes \delta n/2}) + 2^{-\Omega(\delta n)}. \]
\end{lemma}
\begin{proof}
Deferred to~\cref{subsec:rr}.
\end{proof}
\subsubsection{Path Tricks}
\label{overview:pathtrick}
A natural approach towards showing \cref{thm:nocomplex} is to induct on $n$, where the base case $n = 1$ follows exactly because $T$ is nonembeddable over $\D$. However, we don't know how to induct on such a base case. We instead first work towards establishing a base case for functions $F, G, H$ that are $\ell_2$-bounded, instead of $\ell_\infty$-bounded. To achieve this, the idea of \cite{BKM2} was to use a \emph{path trick} (\cref{subsec:pathtrick}), which changes the distribution $\mu$ to a larger distribution $\mu^+$ over $\A' \times \B \times \C$, where $\mu^+_{yz}$ has full support. In fact, we can assume the support is uniform, by applying the random restriction trick (\cref{lemma:rr}).

Formally, the path trick takes an integer $r$, sets $\A^+ \subseteq \A^{2^{r+1}-1}$ and for all $y_1, z_1, y_2, z_2, \dots, y_{2^r}, z_{2^r}$ such that $(y_i, z_i), (z_i, y_{i+1}) \in \supp(\mu_{yz})$, adds the following to $\supp(\mu^+)$. Let $(x_i, y_i, z_i) \in \supp(\mu)$ for $i = 1, \dots, 2^r$, and $(x_i', y_{i+1}, z_i) \in \supp(\mu)$ for $i = 1, 2, \dots, 2^r - 1$.
Then, \[ (\{x_1, x_1', x_2, \dots, x_{2^r}\}, y_1, z_{2^r}) \in \A^+ \times \B \times \C \] is added to $\supp(\mu^+)$.
For sufficiently larger $r$, if $\mu_{yz}$ was initially pairwise-connected, then $\mu'_{yz}$ has full support. Also, we can show that if there were functions $F, G, H$ with \[ \left|\E_{(x,y,z)\sim \mu^{\otimes n}}\left[F(x)G(y)H(z) \prod_{i=1}^n T(x_i,y_i,z_i) \right] \right| \ge \eps, \] then there are functions $\tilde{F}: (\A^+)^n \to \D$, $\tilde{G}: \B^n \to \D$, and $\tilde{H}: \C^n \to \D$ with
\[ \left|\E_{(x,y,z)\sim (\mu^+)^{\otimes n}}\left[\tilde{F}(x)\tilde{G}(y)\tilde{H}(z) \prod_{i=1}^n T^+(x_i,y_i,z_i) \right] \right| \ge \eps^{O_r(1)}, \]
for some tensor $T^+$ that is still nonembeddable over $\D$, i.e., we only incur polynomial loss. We prove this in \cref{lemma:pathtrick}.

At this point, it is worth discussing how the path trick formally helps with enlarging the distribution $\mu$, towards eventually ensuring that $\supp(\mu)$ is all $x+y+z = 0$ for $(x, y, z) \in \cA \times \cA \times \cA$. Towards this, we first calculate how embeddings evolve under the path trick. Given a embedding $(\a, \b, \c)$ of $\mu$, one can easily check that it uniquely induces a embedding on $\mu^+$ of the form $(\a^+, \b, \c)$, where
\[ \a^+(\{x_1, x_1', x_2, \dots, x_{2^r}\}) = \a(x_1) - \a(x_1') + \dots - \a(x_{2^r-1}') + \a(x_{2^r}). \]
Because $(\{x, x, \dots, x\}, y, z) \in \supp(\mu^+)$, and $\a^+(\{x, \dots, x\}) = \a(x)$, we know that the image of $\a^+$ on $\mu$ is nondecreasing. We can also show that by carefully applying a sequence of path tricks, that we can increase its size until the image is a subgroup (\cref{lemma:saturated}).

\subsubsection{Establishing an Inductive Hypothesis}
\label{overview:induct}
When $\mu_{yz}$ is a product distribution, we can establish a weak base case via the Cauchy-Schwarz inequality:
\begin{align} &\left|\E_{(x,y,z)\sim \mu}\left[F(x)G(y)H(z)T(x,y,z) \right] \right| \nonumber \\ \le ~&\left(\E_{(x,y,z)\sim \mu} |F(x)|^2 \E_{(x,y,z)\sim\mu} |G(y)|^2|H(z)|^2 \right)^{1/2} = \|F\|_2 \|G\|_2 \|H\|_2, \label{eq:overviewcs}
\end{align}
where the final equality critically used that $\mu_{yz}$ is a product distribution.
If the right hand side was instead $(1-\delta)\|F\|_2\|G\|_2\|H\|_2$, then we would be able to actually perform the induction on $n$ via singular value decomposition (SVD) on $F, G, H$, thinking of $F: \A^n \to \bbC$ as a matrix in $\bbC^{\A \times \A^{n-1}}$. Thus, we want to say something like: if $T$ is nonembeddable over $\D$, then equality in \eqref{eq:overviewcs} cannot hold.
Equality in \eqref{eq:overviewcs} holds only if $\bar{F(x)} = \theta G(y)H(z)T(x,y,z)$ for some $\theta \in \bbC$, for all $(x, y, z) \in \supp(\mu)$. This looks quite close to $T$ having an embedding over $\D$, except that $F(x), G(y), H(z)$ may be $0$.

In \cite{BKM2}, this issue (which they called the \emph{Horn-SAT obstruction}), was handled as follows. First, perform more path tricks to the distribution to transform it to a larger distribution (which we relabel $\mu$ for convenience) satisfying a \emph{relaxed base case} (see \cref{def:relaxed}). Informally, the distribution $\mu$ over $\A \times \B \times \C$ satisfies the relaxed base case, if for some subset $\A' \subseteq \A$, any function $F$ with nontrivial variance on $\A'$, and $G, H$, satisfy the base case, i.e.,
\[ \left|\E_{(x,y,z)\sim \mu}\left[F(x)G(y)H(z)T(x,y,z) \right]\right| \le (1-\delta)\|F\|_2 \|G\|_2 \|H\|_2. \]
We too need to establish a relaxed 
base case, and we do so by analogous 
ideas.

\subsubsection{High and Low Degree Parts}
 With the relaxed base case in hand, we are able to perform the induction on functions $F$ where during the SVD, the the resulting functions all have nontrivial variance on $\A'$. To handle general functions $F$, we decompose them via an Efron-Stein decomposition (\cref{subsec:es}), into functions that are have equal degrees on $\A'$. The reader can think of an Efron-Stein as a version of a Fourier transform, on a set $\A$ without additive structure. Just as with a Fourier transform, high-degree functions have high influence/variance on each coordinate. Thus, starting from a high-degree function, one show that while performing SVDs in the induction, it remains high-degree for many iterations, allowing us to perform the induction. This allows us to handle the high-degree parts of $F$.

For the lower-degree parts of $F$, their variance on $A'$ is low, so intuitively $F$ does not change much under the noise operator $\TT^{\eff}$ which mixes inputs in $A'$. Applying $\TT^{\eff}$ (and taking random restrictions) makes the input distribution symmetric over $A'$. This case ends up reducing to the two-variable case on $(y, z)$, where an exponential decay bound is known (see \cref{lemma:2dcase}).

\subsubsection{The Case with Complex Embeddings}
The above discussion was for showing \cref{thm:nocomplex}, when $T$ was nonembeddable over $\D$. As discussed in \cref{subsec:master1}, we require an alternate analytic theorem which reduces to the case where $\mu$ is uniform over $x+y+z = 0$, $(x, y, z) \in \cA \times \cA \times \cA$. We describe the idea for adapting to this case. The approach for this case is largely the same as for proving \cref{thm:nocomplex} as discussed above: establish a variant of the relaxed base case, and split into high and low-degree parts. Here, high and low-degree are with respect to the noise operator (which we call $\TT^{\modest})$ which mixes over $x, x'$ with the same master embedding, i.e., $\ama(x) = \ama(x')$. Once again, the contribution of the high-degree parts vanishes, and the low-degree parts intuitively correspond to applying $\TT^{\modest}$. This is exactly what we wanted to show: that $x, x'$ with $\ama(x) = \ama(x')$ are (approximately) indistinguishable from the perspective of strategies $f$, and thus we can reduce the input alphabets to $\cA$ themselves.

\section{Preliminaries}
\label{sec:prelim}

\subsection{General Notation}
\label{subsec:notation}
Let $[n] = \{1, \dots, n\}$. For a real number $x$ we write $\|x\|_{\R/\Z} := \min_{z \in \Z} |x-z|$. We use $a \ls b$ to denote $a = O(b)$ and $a \gs b$ to denote $a = \Omega(b)$. We let $\C$ denote complex numbers, and let $\D$ denote the unit disk. We let $\supp(\mu)$ denote the support of a distribution $\mu$.

\paragraph{Matrices.} We briefly use matrices in a few places, when we are working with ``two-player'' like games, such as merging symbols (\cref{lemma:merging}). We will denote matrices using boldface. For a matrix $\mA$, we let $\mA^{\otimes n}$ denote its $n$-th Kronecker power. For a vector $v$, we let $\diag(v)$ denote the diagonal matrix whose diagonal entries are $v_i$.

\subsection{Efron-Stein Decomposition}
\label{subsec:es}
We will use the standard Efron-Stein decomposition on the product space $L_2(\A^n, \mu_x^{\otimes n})$. For a function $F \in L_2(\A^n, \mu_x^{\otimes n})$ and $S \subseteq [n]$, we let $F^{=S}$ be the homogeneous function of degree $|S|$ that depends on the coordinates in $S$. Let $F^{=i} = \sum_{|S|=i} F^{=S}$.

\paragraph{Effective degree and influences.} In our relaxed base case (\cref{lemma:relaxed}), we will consider the variance on a subset $\A' \subseteq \A$.
Thus, we define \emph{effective} degrees and influences of a function $f$, which are the influence and degrees that $\A'$ contributes.

Create a orthonormal basis $B$ for $L_2(\A, \mu_x)$ of the form $B = B_1 \cup B_2$, where $B_1$ is an orthonormal basis for functions that are constant on $\A'$, and $B_2$ is an orthonormal basis for functions supported on $\A'$ that are orthogonal to the constant function.

For any function $F \in L_2(\A^n, \mu_x^{\otimes n})$ we can write $F = \sum_{\chi \in B^{\otimes n}} \hat{F}(\chi) \chi$, where $\hat{F}(\chi) = \l F, \chi \r$. Let $\chi_{\const}$ be the constant function on $\A$, which we assume is in $B_1$. For a character $\chi \in B^{\otimes n}$, we define its degree and effective degree as
\[ \deg(\chi) = |\{ i \in [n] : \chi_i \neq \chi_{\const}\}| \enspace \text{ and } \enspace \effdeg(\chi) = |\{i \in [n] : \chi_i \in B_2\}|. \]
Note that $\deg(\chi) \ge \effdeg(\chi)$.

This allows us to define the decomposition into effective degree homogeneous parts. Let
\[ F^{\eff= i} = \sum_{\chi\in B^{\otimes n}, \effdeg(\chi) = i} \hat{F}(\chi) \chi. \]

We now define the effective influence.
\[ I_{\eff}[F] = \sum_{i\in[n]} I_{i,\eff}[F] \enspace \text{ for } \enspace I_{i,\eff}[F] = \E_{(x,y) \in \A^n}[|F(x) - F(y)|^2], \]
where $(x,y)$ are sampled as follows. Sample $x$ distributed as $\mu_x^{\otimes n}$, and set $y_j = x_j$ for all $j \neq i$. If $x_i \in \A \backslash \A'$ set $y_i = x_i$, and if $x_i \in \A'$ then sample $y_i \in \A'$ independently.

Analogous to the case of standard influences, we know that
\[ I_{i,\eff}[F] = \sum_{\chi\in B^{\otimes n}, \chi_i \in B_2} |\hat{F}(\chi)|^2 \enspace \text{ and } \enspace I_{\eff}[F] = \sum_{\chi\in B^{\otimes n}} \effdeg(\chi)|\hat{F}(\chi)|^2 = \sum_{0 \le i \le n} i \|F^{\eff=i}\|_2^2. \]

We finally define the effective noise operator, which we denote as $\TT^{\eff}_{1-\rho}$. Define $\TT^{\eff}_{1-\rho}F(x) = \E_{y \sim N^{\eff}_{1-\rho}(x)}[F(y)]$ where $y \sim N^{\eff}_{1-\rho}(x)$ is distributed as follows. For each $i \in [n]$, with probability $1-\rho$ set $y_i = x_i$. With probability $\rho$, if $x_i \in \A \backslash \A'$, then set $y_i = x_i$. Otherwise if $x_i \in \A'$, then $y_i$ is sampled independently according to $\mu_x$ conditioned on being in $\A'$.
As with the standard noise operator, we have
\[ \TT^{\eff}_{1-\rho}F = \sum_{0\le i \le n} (1-\rho)^i F^{\eff=i}. \]

\subsection{Random Restrictions}
\label{subsec:rr}
Consider distributions $\mu, \mu', \nu$ on $\A \times \B \times \C$ with $\mu = \alpha \mu' + (1-\alpha)\nu$. We can take a random restriction of functions $F, G, H$ on $\A^n, \B^n, \C^n$ as follows. Let $I \subseteq [n]$, where each $i \in I$ with probability $1-\alpha$. 
Sample $p := (x(p), y(p), z(p)) \in \nu^{\otimes I}$. We may define the restricted functions $F_{I\to p}: \A^{[n] \setminus I} \to \bbC$ as $F_{I\to p}(x) = F(x, x(p))$, and similar for $G, H$. Note that over the distribution of $(I, p)$, and $x \sim (\mu')^{\otimes [n] \setminus I}$, the distribution of $(x, x(p))$ is identical to $\mu_x^{\otimes n}$. At a high level, this lets us move from a distribution $\mu$ on $[n]$ to a distribution $\mu'$ on about $\alpha n$ coordinates, as long as $\alpha \mu' \le \mu$ pointwise. This was exactly the content of \cref{lemma:rr}, which we now prove.
\begin{proof}[Proof of \cref{lemma:rr}]
Take a random restriction of $\mu$ as described above. The probability that $|[n] \setminus I| < \delta n/2$ is at most $2^{-\Omega(\delta n)}$, by a Chernoff bound. Thus, for a strategy $f, g, h$ achieving $\val(\cG_{\mu}^{\otimes n})$ we know that
\begin{align*}
    \val(\cG_{\mu}^{\otimes n}) &= \Pr_{(x,y,z)\sim\mu^{\otimes n}}\left[f(x)_i + g(y)_i + h(z)_i = t(x_i,y_i,z_i) \enspace \forall \enspace i = 1, \dots, n \right] \\ &= \Pr_{\substack{I \sim_{1-\delta} [n] \\ (x(p),y(p),z(p)) \sim \nu^{\otimes I}}} \Pr_{(x,y,z) \sim (\mu')^{\otimes [n] \setminus I}} \left[ f(x)_i + g(y)_i + h(z)_i = t(x_i,y_i,z_i) \enspace \forall \enspace i = 1, \dots, n \right] \\
    &\le \Pr_{\substack{I \sim_{1-\delta} [n] \\ (x(p),y(p),z(p)) \sim \nu^{\otimes I}}} \Pr_{(x,y,z) \sim (\mu')^{\otimes [n] \setminus I}} \left[ f(x)_i + g(y)_i + h(z)_i = t(x_i,y_i,z_i) \enspace \forall \enspace i \in [n] \setminus I \right] \\
    &\le \E_{I \sim_{1-\delta} [n]} \val(\cG_{\mu'}^{\otimes |I|}) \le \val(\cG_{\mu'}^{\otimes \delta n/2}) + 2^{-\Omega(\delta n)}. \qedhere
\end{align*}
\end{proof}

\section{No Complex Solutions: Reductions}
\label{sec:nocomplex}

The main goal of the following two sections is to show \cref{thm:nocomplex}, restated below.
\nocomplex*

Combining this with \cref{lemma:trivial} directly implies an exponential decay bound on $3$-player $\mathsf{XOR}$ games where the target function $t$ does not have linear structure over any abelian group.
\begin{corollary}
\label{cor:complex}
    Let $\cG$ be a game given by $\mu, t$. If $\mu$ is pairwise-connected, and $T(x,y,z) = \omega^{t(x,y,z)}$ for $\omega = \exp(2\pi i/m)$ is nonembeddable over $\D$ then there is a constant $c := (\cG, \mu) > 0$ with $\val(\cG^{\otimes n}) \le 2^{-cn}$.
\end{corollary}
\begin{proof}
For $S \in \Z_m^n$, let $1(S) = |\{i : S_i = 1\}|$. Then by \cref{thm:nocomplex} and the triangle inequality, the winning probability as written in \cref{lemma:identity} at most
\[ \E_{S \sim \Z_m^n} 2^{-c 1(S)} \le 2^{-c'n}. \qedhere\]
\end{proof}

In the remainder of the section and the next we show \cref{thm:nocomplex}.
Before starting, we make the trivial observation that if $|T(x,y,z)| < 1$ for any $(x,y,z) \in \supp(\mu)$, then the result is clear. So for all $T$ encountered in the series of reductions we do later, we assume $|T(x,y,z)| = 1$ for all $(x,y,z) \in \supp(\mu)$.
\begin{lemma}
\label{lemma:trivial}
If $|T(x,y,z)| < 1$ for some $(x,y,z) \in \supp(\mu)$, then there is $c := c(\mu, T) > 0$ such \[ \left|\E_{(x,y,z)\sim \mu^{\otimes n}}\left[F(x)G(y)H(z) \prod_{i=1}^n T(x_i,y_i,z_i) \right] \right| \le 2^{-cn}. \]
\end{lemma}
\begin{proof}
This follows from the fact that $\E_{(x,y,z)\sim\mu}[|T(x,y,z)|] < 1$.
\end{proof}

\subsection{Path Tricks}
\label{subsec:pathtrick}

We start by defining what it means to do a path trick on $x$. The effect of this will be changing $\A$ to the larger $\A^+$, the distribution to $\mu^+$ over $\A^+ \times \B \times \C$, and the tensor to $t: \supp(\mu^+) \to [-1, 1]$. The goal is to ensure that $\mu^+$ has full support over $(y,z)$.

We define $\mu^+$ as follows.
\begin{definition}[Path trick]
\label{def:pathtrick}
The following procedure takes a distribution $\mu$ on $\A \times \B \times \C$ and generates $\mu^+$ on $\A^+ \times \B \times \C$, which we call applying an ``$x$ path trick with $r$ steps" to $\mu$.
\begin{enumerate}
    \item Sample $y_1 \sim \mu_y$.
    \item Sample $(x_1, z_1)$ from $\mu$, conditioned on $y = y_1$.
    \item Sample $(x_1', y_2)$ from $\mu$, conditioned on $z = z_1$.
    \item Repeat, and in the final step sample $(x_{2^r}, y_{2^{r-1}+1})$ from $\mu$ conditioned on $z = z_{2^{r-1}}$.
\end{enumerate}
Now let $\Sigma^+ \subseteq \Sigma^{2^r-1}$. $\mu^+$ is over the sequences $((x_1, x_1', \dots, x_{2^{r-1}-1}', x_{2^{r-1}}), y_1, z_{2^{r-1}})$ generated by the process above.
\end{definition}
Throughout, we will let $\vec{x} = (x_1, x_1', \dots, x_{2^{r-1}-1}', x_{2^{r-1}})$.

Now we define $T^+(\vec{x}, y, z)$ as follows:
\begin{align*}
    T^+(\vec{x}, y, z) = \mathop{\E}_{\substack{x_1, x_1', \dots, x_{2^{r-1}-1}', x_{2^{r-1}} \\ y_1, \dots, y_{2^{r-1}} \\ z_1, \dots, z_{2^{r-1}}}} \Bigg[ &\prod_{i=1}^{2^{r-1}} T(x_i,y_i,z_i) \prod_{i=1}^{2^{r-1}-1} \bar{T(x_i',y_{i+1}, z_i)} \enspace \Bigg{|} \enspace \\ &\vec{x} = (x_1, x_1', \dots, x_{2^{r-1}-1}', x_{2^{r-1}}), y = y_1, z = z_{2^{r-1}} \Bigg]
\end{align*}
Note that $T^+$ is still nonembeddable over $\D$.
\begin{lemma}
\label{lemma:maintain}
As defined above, if $T$ is nonembeddable over $\D$ then $T^+$ is nonembeddable over $\D$. If $(y,z)$ was pairwise-connected in $\mu$ and $2^{r-1} \ge \min\{|\B|, |\C|\}$, then the support of $(y,z)$ in $\mu^+$ is full.
\end{lemma}
\begin{proof}
Define $(x,\dots,x) = \bar{x} \in \Sigma^+$ for all $x \in \Sigma$. Note that either $T^+(\bar{x}, y, z) = T(x,y,z)$, or else $|T^+(\bar{x},y,z)| < 1$, where we are trivially done by \cref{lemma:trivial}. So $T^+$ is nonembeddable. The second claim on the support of $(y,z)$ in $\mu^+$ is clear.
\end{proof}
Performing the path trick maintains pairwise-connectivity.
\begin{lemma}
\label{lemma:pathconnected}
If $\mu$ is pairwise-connected, then $\mu^+$ is pairwise-connected.
\end{lemma}
\begin{proof}
$(y,z)$-pairwise-connectivity is clear, so we focus on $(x,y)$-pairwise-connectivity by symmetry. We first show that any two symbols $y, y' \in \B$ are connected. This is because $\mu$ is pairwise-connected, and $((x,x,\dots,x), y) \in \supp(\mu^+)$ for any $(x,y) \in \supp(\mu)$. Now, every $\vec{x} \in \A^+$ in adjacent to some $y \in \B$, which completes the proof.
\end{proof}

We now show that performing the path trick can only decrease the value polynomially. 
\begin{lemma}
\label{lemma:pathtrick}
If there are functions $F:\A^n \to \D$, $G:\B^n \to \D$, $H:\C^n \to \D$ with \[ \left| \E_{(x,y,z)\sim\mu^{\otimes n}}[F(x)G(y)H(z) \prod_{i=1}^n T(x_i,y_i,z_i)] \right| \ge \eps, \]
then for $\mu^+, T^+$ as defined above, there are functions $F^+: (\A^{2^r-1})^n \to \D$ and $\wt{H}: \C^n \to \D$ with
\[ \left| \E_{(x,y,z)\sim(\mu^+)^{\otimes n}}[F^+(x)G(y)\wt{H}(z) \prod_{i=1}^n T^+(x_i,y_i,z_i)] \right| \ge \eps^{2^r}. \]
\end{lemma}
\begin{proof}
We start by performing a simple reduction on $H$. Define \[ \wt{H}(z) := \E_{(x,y,z') \sim \mu^{\otimes n}}[\bar{F(x)G(y)T(x,y,z)} | z' = z]. \] Then note that
\begin{align*} \left| \E_{(x,y,z)\sim\mu^{\otimes n}}[F(x)G(y)H(z)T(x,y,z)] \right| &= \left| \E_{z\sim\mu_z}[H(z)\bar{\wt{H}(z)}] \right| \le \|H\|_2 \|\wt{H}\|_2 \le \|\wt{H}\|_2 \\
&= \E_{(x,y,z)\sim\mu^{\otimes n}}[F(x)G(y)\wt{H}(z)T(x,y,z)]^{1/2}. \end{align*}
Hence $\left| \E_{(x,y,z)\sim\mu^{\otimes n}}[F(x)G(y)\wt{H}(z)T(x,y,z)] \right| \ge \eps^2$. The remainder of the proof involves repeated application of Cauchy-Schwarz.

We show the following by induction on $r$:
\begin{align} \left| \mathop{\E}_{\substack{x_1,x_1',\dots,x_{2^{r-1}}, x_{2^{r-1}}' \\ y_1, \dots, y_{2^{r-1}+1} \\ z_1, \dots, z_{2^{r-1}}}}\left[\prod_{i=1}^{2^{r-1}} F(x_i) \bar{F(x_i')} \cdot G(y_1) \bar{G(y_{2^{r-1} + 1})} \prod_{i=1}^{2^{r-1}} T(x_i,y_i,z_i) \bar{T(x_i',y_{i+1},z_i)} \right] \right| \ge \eps^{2^r}. \label{eq:biginduct} \end{align}
Let us establish the base case $r = 1$.
\begin{align*}
\eps^2 &\le \left| \E_{(x,y,z)\sim\mu^{\otimes n}}[F(x)G(y)\wt{H}(z)T(x,y,z)] \right| \\ &= \left| \mathop{\E}_{\substack{z_1 \\ x_1, y_1, x_1', y_2}}[F(x_1) \bar{F(x_1')} G(y_1) \bar{G(y_2)} T(x_1,y_1,z_1) \bar{T(x_1',y_2,z_1)} ] \right|.
\end{align*}
The inductive step will follow by squaring this, and using Cauchy-Schwarz on the variable other than $y_{2^{r-1}+1}$. Precisely,
\begin{align*}
    \eps^{2^{r+1}} &\le \left| \mathop{\E}_{\substack{x_1,x_1',\dots,x_{2^{r-1}}, x_{2^{r-1}}' \\ y_1, \dots, y_{2^{r-1}+1} \\ z_1, \dots, z_{2^{r-1}}}}\left[\prod_{i=1}^{2^{r-1}} F(x_i) \bar{F(x_i')} \cdot G(y_1) \bar{G(y_{2^{r-1}+1})} \prod_{i=1}^{2^{r-1}} T(x_i,y_i,z_i) \bar{T(x_i',y_{i+1},z_i)} \right] \right|^2 \\
    &\le \mathop{\E}_{y_{2^{r-1}+1}} \left| \mathop{\E}_{\substack{x_1,x_1',\dots,x_{2^{r-1}}, x_{2^{r-1}}' \\ y_1, \dots, y_{2^{r-1}} \\ z_1, \dots, z_{2^{r-1}}}}\left[\prod_{i=1}^{2^{r-1}} F(x_i) \bar{F(x_i')} \cdot G(y_1) \prod_{i=1}^{2^{r-1}} T(x_i,y_i,z_i) \bar{T(x_i',y_{i+1},z_i)} \right] \right|^2.
\end{align*}
Note that this last expression exactly equals \eqref{eq:biginduct} for $r+1$, completing the induction. Finally, the definition of $\wt{H}$ gives that the LHS expression in \eqref{eq:biginduct} can be expressed as:
\begin{align*} \left| \mathop{\E}_{\substack{x_1,x_1',\dots,x_{2^{r-1}} \\ y_1, \dots, y_{2^{r-1}} \\ z_1, \dots, z_{2^{r-1}}}}\left[\prod_{i=1}^{2^{r-1}} F(x_i) \prod_{i=1}^{2^{r-1}-1} \bar{F(x_i')} \cdot G(y_1) \wt{H}(z_{2^{r-1}}) \prod_{i=1}^{2^{r-1}} T(x_i,y_i,z_i) \prod_{i=1}^{2^{r-1}-1} \bar{T(x_i',y_{i+1},z_i)} \right] \right|.
\end{align*}
So we can set $F^+(\vec{x}) = \prod_{i=1}^{2^{r-1}} F(x_i) \prod_{i=1}^{2^{r-1}-1} \bar{F(x_i')}$. This is exactly what we wanted to show, once we note that only the $T$ part of the previous expression depends on $y_2, \dots, y_{2^{r-1}}, z_1, \dots, z_{2^{r-1}-1}$.
\end{proof}

\subsection{Merging Symbols}
\label{subsec:merging}
Let us go back to the setting of a distribution $\mu$ over $\A \times \B \times \C$ (you can think of it as if we have already applied a path trick to $x$).
Consider a graph $G$ with vertices in $\A$, where $x, x' \in \A$ are connected if and only if there are $y \in \B, z \in \C$ such that $(x,y,z), (x',y,z) \in \supp(\mu)$.
Give each connected component in $G$ a single representative vertex, and let $\rep(x)$ denote the representative of the component of $x \in \A$. Let $\A^- = \{\rep(x) : x \in \A\}$. The first goal is to reduce to the case where $(\rep(x), y, z) \in \supp(\mu)$.
\begin{lemma}[Merging symbols]
\label{lemma:merging}
If $F:\A^n \to \D$, $G:\B^n \to \D$, $H:\C^n \to \D$ satisfy \[ \left|\E_{(x,y,z) \sim \mu^{\otimes n}}[F(x)G(y)H(z)\prod_{i=1}^n T(x_i,y_i,z_i)] \right| \ge \eps, \]
then there is a distribution $\mu'$ over $\A \times \B \times \C$ with $\{ (\rep(x), y, z) : (x,y,z) \in \supp(\mu) \} \subseteq \supp(\mu')$, nonembeddable $\wt{T}: \supp(\mu') \to \D$, and $\wt{F}: \A^n \to \D$ with
\[ \left|\E_{(x,y,z) \sim {(\mu')}^{\otimes n}}[\wt{F}(x)G(y)H(z)\prod_{i=1}^n \wt{T}(x_i,y_i,z_i)]\right| \ge \eps^{2M}, \]
where $M = |\A|$.
\end{lemma}
\begin{proof}
Define the matrix $\mA \in \bbC^{\A \times (\B \times \C)}$ by $\mA_{(x,(y,z))} = \mu(x,y,z) T(x,y,z)$. Let $\mD_x = \diag(\mu_x)$ and $\mD_{yz} = \diag(\mu_{y,z})$. Define the vector $R(y,z) = G(y)H(z)$.

In the calculations below, to simplify notation we by abuse of notation let $\mA, \mD_x, \mD_{yz}$ denote $\mA^{\otimes n}, \mD_x^{\otimes n}, \mD_{yz}^{\otimes n}$ respectively.
By the hypothesis and Cauchy-Schwarz we get:
\[ \eps \le |\bar{F}^* \mA R| \le \|\mD_x^{1/2} F\|_2 \|\mD_x^{-1/2} \mA R\|_2 \le \|\mD_x^{-1/2} \mA R\|_2, \] so expanding gives
\[ (\mD_{yz}^{1/2}R)^* (\mD_{yz}^{-1/2} \mA^* \mD_x^{-1} \mA \mD_{yz}^{-1/2}) \mD_{yz}^{1/2}R \ge \eps^2. \]
Using $\|\mD_{yz}^{1/2}R\|_2 \le 1$ and repeated Cauchy-Schwarz gives us
\begin{align*} (\mD_{yz}^{1/2}R)^* (\mD_{yz}^{-1/2} \mA^* \mD_x^{-1} \mA \mD_{yz}^{-1/2})^M \mD_{yz}^{1/2}R \ge \eps^{2M}.  \end{align*} Thus for $\wt{F} = \mD_x^{-1} \mA R$ we know that
\begin{align} \wt{F}^* \mA(\mD_{yz}^{-1} \mA^* \mD_x^{-1} \mA)^{M-1} R \ge \eps^{2M}. \label{eq:powered} 
\end{align}
Now we interpret this to extract the desired $\wt{T}$ and $\mu'$. Consider a random walk between $\A$ and $(\B \times \C)$ as follows, where each step is sampled proportional to $\mu$. The total number of steps is $2M-1$, starting from $(y_1,z_1) \sim \mu_{y,z}$. It goes through vertices $(y_1,z_1), x_1, (y_2,z_2), \dots, (y_M, z_M), x_M$.
Then we can set $\mu'(x,y,z) = \Pr[x = x_M, y = y_1, z = z_1]$ and 
\[ \wt{T}(x,y,z) = \E\left[\prod_{i=1}^M T(x_i,y_i,z_i) \prod_{i=1}^{M-1} \bar{T(x_i, y_{i+1}, z_{i+1})} \enspace \Bigg| \enspace x = x_M, y = y_1, z = z_1 \right]. \]
Then \eqref{eq:powered} exactly tells us that $\mu'$, $\wt{T}$ satisfy the conclusion of the lemma.
To check that $\wt{T}$ is nonembeddable note that $\wt{T}(x,y,z) = T(x,y,z)$ for all $(x,y,z) \in \supp(\mu)$, or $|\wt{T}(x,y,z)| < 1$, where we are done by \cref{lemma:trivial}.
\end{proof}
Now, by the random restriction trick (\cref{lemma:rr}), we can ensure that $\supp(\mu^-)$ contains exactly $(\rep(x), y, z)$.
\begin{lemma}
\label{lemma:merging2}
Under the hypotheses of \cref{lemma:merging}, there is a distribution $\mu^-$ with $\supp(\mu^-) = \{ (\rep(x), y, z) : (x,y,z) \in \supp(\mu) \}$, $n' \ge cn$, and functions $\wt{F}: (\A^-)^{n'} \to \D, \wt{G}: \B^{n'} \to \D, \wt{H}: \C^{n'} \to \D$ such that
\[ \left|\E_{(x,y,z) \sim {(\mu^-)}^{\otimes n'}}[\wt{F}(x)\wt{G}(y)\wt{H}(z)\prod_{i=1}^{n'} T(x_i,y_i,z_i)]\right| \ge \eps^{2M}/2. \]
\end{lemma}
We have reduced the size of the alphabet and support of the distribution so we may worry that $T: \supp(\mu^-) \to \{-1, 1\}$ may become embeddable into $\D$. We show that if this is the case, then we can show that the game has exponentially decreasing value. This can be viewed as an instance of two player games.
\begin{lemma}
\label{lemma:validt}
    If $T: \supp(\mu^-) \to \D$ is embeddable over $\D$ then for some $c := c(T,\mu) > 0$, \[ \left|\E_{(x,y,z)\sim\mu^{\otimes n}}\left[F(x)G(y)H(z)\prod_{i=1}^n T(x_i,y_i,z_i)\right]\right| \le 2^{-cn}. \]
\end{lemma}
\begin{proof}
We first show that if $T: \supp(\mu^-) \to \D$ is embeddable, then there are no functions $A: \A \to \D$ and $B: (\B \times \C) \to \D$ for which $A(x) B(y,z) = T(x,y,z)$ for all $(x,y,z) \in \supp(\mu)$. Assume such $A, B$ existed. We will show how to create an embedding on $\supp(\mu)$ from one on $\supp(\mu^-)$. Take an embedding $A'(x)B'(y)C'(z) = T(x, y, z)$ for all $(x, y, z) \in \supp(\mu^-)$. Then we have
\[ T(x, y, z) = A(x) B(y, z) = \frac{A(x)}{A(\rep(x))} T(\rep(x), y, z) = \frac{A(x)A'(\rep(x))}{A(\rep(x))} B'(y) C'(z), \] hence $T$ admits an embedding on $\mu$, a contradiction.

So we may assume that $T: \A \times (\B \times \C) \to \D$ is nonembeddable. Critically, we have combined $\B, \C$ so that this is two dimensional now. Thus, the conclusion follows from the two-dimensional case, in \cref{lemma:2dcase}.
\end{proof}

\subsection{The Relaxed Base Case}
\label{subsec:relaxed}

In this section, we reduce to $\mu$ which satisfy a certain \emph{relaxed base case}. This is to handle the \emph{Horn SAT obstruction}, as introduced in \cite{BKM2}.
\begin{definition}[Relaxed base case]
\label{def:relaxed}
We say that a distribution $\mu$ on $\A \times \B \times \C$ and $\A' \subseteq \A$ satisfy the relaxed base case if:
\begin{itemize}
    \item $\mu_{y,z}$ is a product distribution.
    \item $(y,z)$ uniquely determines $x$ in $\supp(\mu)$.
    \item For a nonembeddable tensor $T$, functions $G: \B \to \bbC$, $H: \C \to \bbC$, and $F: \A \to \bbC$ with \[ I_{\eff}[F(x)] = \E_{x,x'\sim\mu_x}[1_{x,x' \in \A'} |F(x) - F(x')|^2] \ge \tau \|F\|_2^2, \]
    we have for a constant $c := c(T,\mu)$ and $M = \max\{|\A|, |\B|, |\C|\}$,
    \[ \left|\E_{(x,y,z)\sim\mu}\left[F(x)G(y)H(z)T(x,y,z) \right]\right| \le (1 - c\tau^{100M})\|F\|_2\|G\|_2\|H\|_2. \]
\end{itemize}
\end{definition}
The following lemma gives a reduction to the relaxed base case.
\begin{lemma}[Reduction to relaxed base case]
\label{lemma:relaxed}
If there are functions $F:\A^n \to \D$, $G:\B^n \to \D$, $H:\C^n \to \D$ with \[ \left| \E_{(x,y,z)\sim\mu^{\otimes n}}[F(x)G(y)H(z) \prod_{i=1}^n T(x_i,y_i,z_i)] \right| \ge \eps, \]
then we can find $\A^+, \B^+, \C$ and $\mu^+$ on $\A^+ \times \B^+ \times \C$ and $\A' \subseteq \A^+$ satisfying the relaxed base case, and $n' \ge cn$, functions $\wt{F}: (\A^+)^{n'} \to \D$, $\wt{G}: (\B^+)^{n'} \to \D$, $\wt{H}: \C^{n'} \to \D$, and nonembeddable $\wt{T}: \supp(\mu^+) \to \D$ with
\begin{align} \left|\E_{(x,y,z)\sim(\mu^+)^{\otimes n'}}[\wt{F}(x)\wt{G}(y)\wt{H}(y)\prod_{i=1}^{n'}\wt{T}(x_i,y_i,z_i)] \right| \ge \eps^{O_M(1)}. \label{eq:reductionresult} \end{align}
\end{lemma}
\begin{proof}
We construct $\mu^+$ as follows starting from $\mu$.
\begin{enumerate}
    \item Apply the path trick (\cref{lemma:pathtrick}) to $y$. The result is a distribution $\mu^{(1)}$ on $\A \times \B^+ \times \C$ such that $\mu^{(1)}_{x,z}$ has full support.
    \item Apply the path trick (\cref{lemma:pathtrick}) to $x$. The result is a distribution $\mu^{(2)}$ on $\A^+ \times \B^+ \times \C$ such that $\mu^{(2)}_{y,z}$ is full. We also will require that for all $z$ and $\vec{x} = (x,\dots,x) \in \A^+$, there is a $\vec{y} \in \B^+$ with $(\vec{x}, y, z) \in \supp(\mu^{(2)})$. This follows because if $(x, \vec{y}, z) \in \supp(\mu^{(1)})$, then $(\vec{x}, \vec{y}, z) \in \supp(\mu^{(2)})$.
    \item Merge symbols in $\A^+$ using \cref{lemma:merging,lemma:merging2}. Relabel the alphabet as $\A^+$ still. Call this distribution $\mu^{(3)}$.
    \item Apply a random restriction so that $(y,z)$ is uniform. This is our distribution $\mu'$.
    \item Set $\A' = \{\rep((x,\dots,x)) : x \in \A\}$.
\end{enumerate}
The existence of $\wt{F},\wt{G},\wt{H}$, and nonembeddable $\wt{T}$ satisfying \eqref{eq:reductionresult} follows from combining \cref{lemma:pathtrick,lemma:merging,lemma:merging2,lemma:validt}.
It suffices to show that $\mu'$ satisfies the relaxed base case. The first two properties follow by construction. We proceed to showing the third property.

In fact we only use the weaker property that $\E_{x\sim\mu_x}[1_{x\in\A'} |F(x)|^2] \ge \tau\|F\|_2^2$. Assume that $\|F\|_2 = \|G\|_2 = \|H\|_2 = 1$. For simplicity of notation, we use $\ls, \gs$ in this proof to suppress constants depending on $M, \mu$.

Let us multiply $F(x)$ pointwise by the same unit complex number so that 
\[ \E_{(x,y,z)\sim\mu}\left[F(x)G(y)H(z)T(x,y,z)\right] \] is a positive real number. We may assume this throughout. Then we have the equality
\[ \left|\E_{(x,y,z)\sim\mu}\left[F(x)G(y)H(z)T(x,y,z) \right]\right| = 1 - \frac12\E_{(x,y,z)\sim\mu}\left[|\bar{F(x)} - G(y)H(z)T(x,y,z)|^2 \right]. \]
Assume for contradiction that $|\bar{F(x)} - G(y)H(z)T(x,y,z)|^2 \le c^2\tau^{100M}$ for all  $(x,y,z) \in \supp(\mu)$, for some constant $c$. In particular, $||F(x)| - |G(y)||H(z)|| \le c\tau^{50M}$.

Now, there is $x \in \A'$ with $|F(x)|^2 \gs \tau$. For all $z \in \C$ there is $y \in \B^+$ with $(x,y,z) \in \supp(\mu)$. Using that $|G(y)| \ls 1$ for all $y$ (as $\|G\|_2 = 1$), we deduce that
\[ |H(z)| \ge \frac{|F(x)| - c\tau^{50M}}{|G(y)|} \gs \sqrt{\tau}. \]
We next establish that $|G(y)|, |F(x)| \gs \tau^{M/2}$ for all $x \in \A^+, y \in \B^+$.
Indeed, for any $(x, y)$ such that there is a $z \in \C$ with $(x,y,z) \in \supp(\mu)$, then $|F(x)| \ge |G(y)||H(z)| - c\tau^{50M} \gs |G(y)|\sqrt{\tau} - c\tau^{50M}$, and
\[ |G(y)| \ge \frac{|F(x)| - c\tau^{50M}}{|H(z)|} \gs |F(x)| - c\tau^{50M}, \] as $|H(z)| \ls 1$. This implies that $|G(y)|, |F(x)| \gs \tau^{M/2}$ for all $x \in \A^+, y \in \B^+$ because the $x,y$ coordinates are pairwise connected, so there is a path from the $x \in \A'$ with $|F(x)| \gs \tau^{1/2}$ to any $x \in \A^+$ of length at most $M$.

Now, because $T$ is nonembeddable, there is $(x,y,z) \in \supp(\mu)$ with
\[ \left|T(x,y,z) - \frac{\bar{F}(x)}{|\bar{F}(x)|} \cdot \frac{\bar{G}(y)}{|\bar{G}(y)|} \cdot \frac{\bar{H}(z)}{|\bar{H}(z)|}\right| \gs 1. \]
From this, it is direct that
\[ |T(x,y,z)G(y)H(z) - \bar{F(x)}| \gs \min\{|G(y)H(z)|, |F(x)| \} \gs \tau^M. \]
\end{proof}

\subsection{Reducing to High-Degree Functions}
\label{subsec:highdegree}

In this section, we will deduce \cref{thm:nocomplex} from the following bound on functions with high effective degrees, which we show in \cref{sec:induction}.
\begin{restatable}{theorem}{highdegree}
\label{thm:highdegree}
If $\mu$ satisfies the relaxed base case (\cref{def:relaxed}), then there are constants $c := c(T, \mu), C := C(T,\mu) > 0$ such that for all $G: \B^n \to \bbC, H: \C^n \to \bbC$, and $F: \A^n \to \bbC$ that is effectively homogeneous with effective degree $d$,
\[ \left|\E_{(x,y,z)\sim\mu^{\otimes n}}[F(x)G(y)H(z) \prod_{i=1}^n T(x_i,y_i,z_i)]\right| \le 2^{-c \cdot d(d/n)^C} \|F\|_2 \|G\|_2 \|H\|_2. \]
\end{restatable}
To deduce \cref{thm:nocomplex} from 
\cref{thm:highdegree}, the idea is to decompose $F = \TT^{\eff}_{I-\rho}F + (1 - \TT^{\eff}_{1-\rho})F$, where $I$ is the identity operator. We can handle the former term because it causes the distribution $\mu$ to mix over $\A'$, at which point the underlying game becomes essentially connected. The second term has only high-degree terms, and we apply \cref{thm:highdegree}. To achieve the desired quantitative dependence on $n$ in \cref{thm:nocomplex}, we will use a combination of noise operators instead of a single operator to more effectively eliminate low-degree terms.

Our combination of noise operators comes from the following lemma.
\begin{lemma}[Polynomial approximation]
\label{lemma:polyapprox}
For any $d \in \Z_{\ge0}$ and $\eps \in (0,1/2)$, there are $\rho_0, \dots, \rho_d \in [0, 1-\eps]$ and $c_0, \dots, c_d \in \R$ such that:
\begin{enumerate}
    \item $\sum_{j=0}^d c_j\rho_j^k = 1$ for $k = 0, 1, \dots, d$, \label{item:degree<d}
    \item $0 \le \sum_{j=0}^d c_j \rho_j^k \le 1$ for $k \ge d+1$. \label{item:degree>d}
    \item $\sum_{j=0}^d |c_j| \le 2^{O(d\sqrt{\eps})}$. \label{item:totalc}
\end{enumerate}
\end{lemma}
\begin{proof}
The proof uses Chebyshev polynomials and other tricks, and is deferred to \cref{app:poly}.
\end{proof}
We also prove that applying a single noise operator $\TT^{\eff}_{1-\xi}$ to $F$ leads to a game whose value is exponentially decaying.
\begin{lemma}
\label{lemma:applynoise}
For $F: \A^n \to \D$, $G: \B^n \to \D$, $H: \C^n \to \D$, $\mu$ constructed as in \cref{lemma:relaxed}, and any $\xi > 0$, there is a constant $c := c(T,\mu) > 0$ such that
\[ \left|\E_{(x,y,z)\sim\mu^{\otimes n}}[\TT^{\eff}_{1-\xi}F(x) G(y) H(z) \prod_{i=1}^n T(x_i,y_i,z_i) ]\right| \le 2^{-c\xi n}. \]
\end{lemma}
\begin{proof}
Recall that $(y, z)$ determines $x$ in the support of $\mu$. Thus we can (by abuse of notation) write $T(y, z) = T(x, y, z)$.
Let $\mu'$ be the distribution defined as follows: sample $(x,y,z) \sim \mu$, then $x' \sim N^{\eff}_{1-\xi}(x)$, and return $(x',y,z)$.
Then by definition, 
\begin{align*} &\left|\E_{(x,y,z)\sim\mu^{\otimes n}}[\TT^{\eff}_{1-\xi}F(x) G(y) H(z) \prod_{i=1}^n T(y_i,z_i) ]\right| \\ = ~&\left|\E_{(x,y,z)\sim(\mu')^{\otimes n}}[F(x) G(y) H(z) \prod_{i=1}^n T(y_i,z_i) ]\right|. \end{align*}
For sufficiently small constant $c := c(T,\mu)$, we can write $\mu' = c\xi\bar{\mu} + (1-c\xi)\nu$ for the distribution $\bar{\mu}$ which is sampled as follows: sampled $(x,y,z) \sim \mu$ conditioned on $x \in \A'$, and then resample $x \in \A'$ independently.
Taking a random restriction (\cref{lemma:rr}) gives that it suffices to show that for all $F: \A^{n'} \to \D, G: \B^{n'} \to \D, H: \C^{n'} \to \D$, we have
\[ \E_{(x,y,z)\sim(\bar{\mu})^{\otimes n'}}[F(x)G(y)H(z) \prod_{i=1}^{n'} T(y_i, z_i)] \le 2^{-cn'}. \]
Note that $\bar{\mu}$ is connected, so it suffices to show that $T$ is nonembeddable over $\supp(\bar{\mu})$. This follows from the fact that \cref{lemma:merging,lemma:validt} only use values $T(x,y,z)$ for $x \in \A'$ to show nonembeddability of $T$.
\end{proof}
We can now show our main analytic result, \cref{thm:nocomplex}.
\begin{proof}[Proof of \cref{thm:nocomplex}]
By combining all our reductions in \cref{lemma:pathtrick,lemma:merging,lemma:merging2,lemma:relaxed} (which either increase $\eps$ by polynomial factors, or decrease $n$ by a constant), we assume that $\mu$ satisfies the relaxed base case.
Let $I$ be the identity operator on $L_2(\A^n, \mu_x^{\otimes n})$. For parameters $d \in \Z_{\ge0}$ (chosen later) and $\eps = 1/4$, let $\{c_i\}, \{\rho_i\}$ be as in \cref{lemma:polyapprox}, and define $\TT = \sum_{j=0}^d c_j \TT^{\eff}_{\rho_j}.$
We now write
\begin{align*} 
&\E_{(x,y,z)\sim\mu^{\otimes n}}[F(x)G(y)H(z) \prod_{i=1}^n T(x_i,y_i,z_i)] \\
~& \E_{(x,y,z)\sim\mu^{\otimes n}}[\TT F(x)G(y)H(z) \prod_{i=1}^n T(x_i,y_i,z_i)] + \E_{(x,y,z)\sim\mu^{\otimes n}}[(I-\TT)F(x)G(y)H(z) \prod_{i=1}^n T(x_i,y_i,z_i)].
\end{align*}
We handle the first term using \cref{lemma:applynoise}. Indeed,
\begin{align*}
    &\left|\E_{(x,y,z)\sim\mu^{\otimes n}}[\TT F(x)G(y)H(z) \prod_{i=1}^n T(x_i,y_i,z_i)] \right| \\ \le ~&\sum_{j=0}^d |c_j| \left|\E_{(x,y,z)\sim\mu^{\otimes n}}[\TT^{\eff}_{\rho_j}F(x)G(y)H(z) \prod_{i=1}^n T(x_i,y_i,z_i)] \right| \le 2^{O(d\sqrt{\eps})} 2^{-c\eps n},
\end{align*}
by \cref{lemma:applynoise} and that $\rho_j \le 1-\eps$ for all $j$ by \cref{lemma:polyapprox}. We use \cref{thm:highdegree} to handle the second term. Indeed,
\begin{align*}
    &\left|\E_{(x,y,z)\sim\mu^{\otimes n}}[(I-\TT)F(x)G(y)H(z) \prod_{i=1}^n T(x_i,y_i,z_i)] \right| \\
    = ~& \left| \sum_{k=0}^n (1-\sum_{j=0}^d c_j \rho_j^k) \E_{(x,y,z)\sim\mu^{\otimes n}}[F^{\eff=k}(x)G(y)H(z) \prod_{i=1}^n T(x_i,y_i,z_i)] \right| \\
    \le ~& \sum_{k=d+1}^n \left|\E_{(x,y,z)\sim\mu^{\otimes n}}[F^{\eff=k}(x)G(y)H(z) \prod_{i=1}^n T(x_i,y_i,z_i)] \right| \le \sum_{k=d+1}^n 2^{-ck(k/n)^C} \le n \cdot 2^{-d(d/n)^C},
\end{align*}
where the first inequality uses \cref{item:degree<d,item:degree>d} of \cref{lemma:polyapprox}, and the second inequality uses \cref{lemma:highdegree} and $\|F^{\eff=k}\|_2 \le \|F\|_2 \le 1$ for all $k$. To conclude, set $d = c'n$ for sufficiently small constant $c'$, so that the total is at most $2^{O(d\sqrt{\eps})} 2^{-c\eps n} + n \cdot 2^{-d(d/n)^C} \le 2^{-c''n}$.
\end{proof}

\section{Main Induction: No Complex Embeddings}
\label{sec:induction}
The goal of this section is to show \cref{thm:highdegree} by induction on $n$.
\highdegree*

Let $\cP^{(n,d)}$ be the class of all functions $F:\A^n \to \bbC$ such that all monomials in $F$ have effective degree at least $d$. Define
\[ \gamma_{n,d} := \mathop{\sup}_{\substack{\|F\|_2 \le 1, F \in \cP^{(n,d)} \\ \|G\|_2 \le 1, \|H\|_2 \le 1}} \left|\E_{(x,y,z)\sim\mu^{\otimes n}}[F(x)G(y)H(z)\prod_{i=1}^n T(x_i,y_i,z_i)]\right|. \]
\cref{thm:highdegree} follows from the following lemma.
\begin{lemma}
\label{lemma:highdegree}
If $\mu$ satisfies the relaxed base case, then $\gamma_{n,d} \le (1 - c(d/n)^{O_M(1)})\gamma_{n-1,d-1}$.
\end{lemma}
\begin{proof}[Proof of \cref{thm:highdegree}]
Iterating \cref{lemma:highdegree} a total of $d/2$ times gives:
\[ \gamma_{n,d} \le (1 - c(d/n)^{O_M(1)})^{d/2} \le \exp(-c \cdot d(d/n)^{O_M(1)}). \qedhere\]
\end{proof}
The proof of \cref{lemma:highdegree} proceeds by performing an SVD on the functions $F, G, H$.
\begin{lemma}[SVD]
\label{lemma:svd}
For a function $F: \A^n \to \bbC$ with $\|F\|_2 = 1$, for $M = |\A|$, there are functions $F_1, \dots, F_M: \A^{n-1} \to \bbC$, $F_1', \dots, F_M': \A \to \bbC$, and nonnegative real numbers $\aa_1, \dots, \aa_M$ satisfying:
\begin{enumerate}
    \item $F_1, \dots, F_M$ form an orthonormal basis.
    \item $F_1', \dots, F_M'$ form an orthonormal basis.
    \item $\sum_{r=1}^M \aa_r^2 = 1$.
    \item $F(x) = \sum_{r=1}^M \aa_r F_r(x_1,\dots,x_{n-1}) F_r'(x_n)$ for all $x \in \A^n$.
\end{enumerate}
\end{lemma}
Note that each $F_r \in \cP^{(n-1,d-1)}$ because each $F_r'$ has degree $1$.

For notational simplicity, let $I = \{1, \dots, n-1\}$ and $J = \{n\}$. We write $x_I = (x_1, \dots, x_{n-1})$ and $x_J = x_n$. Towards showing \cref{lemma:highdegree}, because $F$ has all monomials of effective degree at least $d$, we know that its effective influence satisfies $I_{\eff}[F] \ge d$. Hence there is an index $j \in [n]$ with $I_{j,\eff}[F] \ge d/n$.
Without loss of generality, we let $j = n$. Let $\tau = d/n$.

This actually implies that the average effective variance of the $F_r'$ is large.
\begin{lemma}
\label{lemma:variance}
We have: $\sum_{r\in[M]} \aa_r^2 I_{\eff}[F_r'] = I_{n,\eff}[F] \ge \tau$.
\end{lemma}
\begin{proof}
Indeed, if $(x_n,y_n)$ are sampled by the distribution used to define $I_{\eff}$, we know
\begin{align*}
I_{n,\eff}[F] &= \E_{x_I,(x_n,y_n)}[|F(x_I,x_n) - F(x_I,y_n)|^2] = \E_{x_I,(x_n,y_n)}\left[\left|\sum_{r\in[M]} \aa_r F_r(x_I)(F_r'(x_n) - F_r'(y_n)) \right|^2\right] \\
&= \E_{x_I,(x_n,y_n)}\left[\sum_{r_1,r_2\in[M]} \aa_{r_1}\aa_{r_2}F_{r_1}(x_I)\bar{F_{r_2}(x_I)} (F_{r_1}'(x_n) - F_{r_1}'(y_n))\bar{(F_{r_2}'(x_n) - F_{r_2}'(y_n))}\right] \\
&= \E_{(x_n,y_n)}\left[\sum_{r\in[M]} \aa_r^2 |F_r'(x_n) - F_r'(y_n)|^2 \right] = \sum_{r\in[M]} \aa_r^2 I_{\eff}[F_r']
\end{align*}
\end{proof}

Perform SVDs on $F, G, H$ using \cref{lemma:svd}. We write
\[ F(x) = \sum_{r=1}^M \aa_rF_r(x_I)F_r'(x_J), G(y) = \sum_{s=1}^M \bb_sG_s(y_I)G_s'(y_J), H(z) = \sum_{t=1}^M \cc_tH_t(x_I)H_t'(x_J). \]
Define the quantities:
\begin{align*} \hat{F_r}(s,t) &:= \E_{(x,y,z)\sim\mu^{\otimes I}}[F_r(x)G_s(y)H_t(z)\prod_{i\in I}T(x_i,y_i,z_i)], \\
\hat{F_r'}(s,t) &:= \E_{(x,y,z)\sim\mu}[F_r'(x)G_s'(y)H_t'(z)T(x,y,z)].
\end{align*}
Then by definition we know that
\[ \E_{(x,y,z)\sim\mu^{\otimes n}}[F(x)G(y)H(z)\prod_{i=1}^n T(x_i,y_i,z_i)] = \sum_{r,s,t\in[M]} \aa_r\bb_s\cc_t\hat{F_r}(s,t)\hat{F_r'}(s,t). \]

\subsection{Finding a Singular Value Gap}
\label{subsec:gap}
Let $\delta < \Delta$ satisfy that none of the $\aa_r, \bb_s, \cc_t$ lie in $[\delta,\Delta]$, and $\delta \le \left(\frac{\Delta}{CM\tau}\right)^{1000M}$.
We can find such $\delta, \Delta$ with $\delta \ge \tau^{O_M(1)}$. Let $R = \{r : \aa_r \ge \Delta\}, S = \{s : \bb_s \ge \Delta\}, T = \{t : \cc_t \ge \Delta\}$. Note that because $|\hat{F_r}(s,t)| \le \gamma_{n-1,d-1}, |\hat{F_r'}(s,t)| \le 1$ we know
\begin{align} \left|\sum_{r,s,t\in[M]} \aa_r\bb_s\cc_t\hat{F_r}(s,t)\hat{F_r'}(s,t)\right| \le \left|\sum_{r\in R,s\in S,t\in T} \aa_r\bb_s\cc_t\hat{F_r}(s,t)\hat{F_r'}(s,t)\right| + \delta M^3 \gamma_{n-1,d-1}. \label{eq:remove} \end{align}

\subsection{Cauchy-Schwarz and Simple Observations}
\label{subsec:firstcs}
For the remainder of the section, we assume that \cref{lemma:highdegree} is false for the choice of $F, G, H$. By the Cauchy-Schwarz inequality we get:
\begin{align} \left|\sum_{r\in R,s\in S,t\in T} \aa_r\bb_s\cc_t\hat{F_r}(s,t)\hat{F_r'}(s,t)\right| \le \sqrt{\sum_{s\in S,t\in T} \bb_s^2 \cc_t^2 \sum_{r\in R} |\hat{F_r}(s,t)|^2}\sqrt{\sum_{r\in R} \aa_r^2 \sum_{s\in S,t\in T} |\hat{F_r'}(s,t)|^2 }. \label{eq:cs} \end{align}
We have the following simple observations:
\begin{lemma}
\label{lemma:bound1}
For all $s, t$ we have: $\sum_{r\in[M]} |\hat{F_r}(s,t)|^2 \le \gamma_{n-1,d-1}^2$.
\end{lemma}
\begin{proof}
    Write \[ \wt{F}_{s,t} = \frac{\sum_{r\in[M]} \bar{\hat{F_r}(s,t)}F_r}{\sqrt{\sum_{r\in[M]} |\hat{F_r}(s,t)|^2}}. \] Clearly $\|\wt{F}_{s,t}\|_2 = 1$. Thus, we get \[ \sum_{r\in[M]} |\hat{F_r}(s,t)|^2 = \left|\E_{(x,y,z)\sim\mu^{\otimes I}}\left[\wt{F}_{s,t}(x)G_r(y)H_s(y) \prod_{i\in I} T(x_i,y_i,z_i) \right] \right|^2 \le \gamma_{n-1,d-1}^2. \]
\end{proof}
Analogously, we have:
\begin{lemma}
\label{lemma:bound21}
For all $s, t$ we have: $\sum_{r\in[M]} |\hat{F_r'}(s,t)|^2 \le 1$.
\end{lemma}
\begin{proof}
The proof is identical to that of \cref{lemma:bound1}.
\end{proof}
\begin{lemma}
\label{lemma:bound2}
For all $r$ we have: $\sum_{s,t\in[M]} |\hat{F_r'}(s,t)|^2 \le 1$.
\end{lemma}
\begin{proof}
This is one of the few places in the proof we use that $x$ is determined by $(y,z)$ under the relaxed based case (\cref{def:relaxed}). To clarify this throughout this proof, we will write $x$ as $x(y,z)$. Define the functions $U_{s,t}: \B \times \C \to \bbC$ as $U_{s,t}(y,z) = G_s'(y)H_t'(z)T(x(y,z),y,z)$. These functions are orthogonal over $L_2(\B\times\C,\mu_{y,z})$. Indeed for $(s_1,t_1)$ and $(s_2,t_2)$,
\begin{align*} \E_{(y,z)\sim\mu_{y,z}}[U_{s_1,t_1}(y,z) \bar{U_{s_2,t_2}(y,z)}] &= \E_{(y,z)\sim\mu_{y,z}}[G_{s_1}'(y)\bar{G_{s_2}'(y)} H_{t_1}'(z) \bar{H_{t_2}'(z)} |T(x(y,z),y,z)|^2] \\
&= \E_{(y,z)\sim\mu_{y,z}}[G_{s_1}'(y)\bar{G_{s_2}'(y)} H_{t_1}'(z) \bar{H_{t_2}'(z)}] = 1_{(s_1,t_1) = (s_2,t_2)},
\end{align*}
where we have used that $\mu_{y,z}$ is a product distribution, $|T(x,y,z)| = 1$ for all $(x,y,z)\in\supp(\mu)$, and orthogonality of the families $G_s'$ and $H_t'$. The lemma follows because $\hat{F_r'}(s,t) = \l F_r', \bar{U_{s,t}} \r$.
\end{proof}
Note now that these lemmas show:
\[ \sum_{s\in S,t\in T} \bb_s^2 \cc_t^2 \sum_{r\in R} |\hat{F_r}(s,t)|^2 \le \gamma_{n-1,d-1}^2 \enspace \text{ and } \enspace \sum_{r\in R} \aa_r^2 \sum_{s\in S,t\in T} |\hat{F_r'}(s,t)|^2 \le 1. \]
Thus, if \cref{lemma:highdegree} fails, we must have near-equality of \cref{lemma:bound1,lemma:bound2}.
\begin{lemma}
\label{cor:bound1}
For all $s\in S,t\in T$ we have: $\sum_{r\in R} |\hat{F_r}(s,t)|^2 \ge (1 - \frac{4\delta M^3}{\Delta^4})\gamma_{n-1,d-1}^2$.
\end{lemma}
\begin{proof}
If the contrary holds, then
\[ \sum_{s\in S,t\in T} \bb_s^2 \cc_t^2 \sum_{r\in R} |\hat{F_r}(s,t)|^2 \le (1 - 4\delta M^3)\gamma_{n-1,d-1}^2, \] because $\bb_s,\cc_t \ge \Delta$. Combining this with \eqref{eq:remove} and \eqref{eq:cs} would show \cref{lemma:highdegree}. 
\end{proof}
\begin{lemma}
\label{cor:bound2}
For all $r\in R$ we have: $\sum_{s\in S,t\in T} |\hat{F_r'}(s,t)|^2 \ge 1 - \frac{4\delta M^3}{\Delta^2}$.
\end{lemma}
\begin{proof}
If the contrary holds, then
\[ \sum_{r\in R} \aa_r^2 \sum_{s\in S,t\in T} |\hat{F_r'}(s,t)|^2 \le 1 - 4\delta M^3, \]
because $\aa_r \ge \Delta$. Combining this with \eqref{eq:remove} and \eqref{eq:cs} would show \cref{lemma:highdegree}. 
\end{proof}
The remainder of the section is split into two cases depending on whether $|R| < |S||T|$ or $|R| \ge |S||T|$. The former is handled by showing that we can perturb $F_r,G_s,H_t$ to reach a contradiction. The latter is handled by using the one-dimensional bound guaranteed by the relaxed base case.

\subsection{Handling the Case \texorpdfstring{$|R| < |S||T|$}{r<st}}
\label{subsec:r<st}
In this case we will show how to build functions $\wt{F}, \wt{G}, \wt{H}$ that violate the inductive hypothesis. Define the $|R|$ dimensional vectors $V^{(s,t)}$, defined as $V^{(s,t)}_r = \hat{F_r}(s,t)$. Because $|R| < |S||T|$, we can find some $(s_1,t_1)$ and $(s_2,t_2)$ such that $|\l V^{(s_1,t_1)}, V^{(s_2,t_2)}\r|$ is non-negligible.
\begin{lemma}
\label{lemma:inner}
For $q > d$, and vectors $v_1,\dots,v_q \in \bbC^d$, there is $i \neq j$ such that \[ |\l v_i, v_j \r| \ge d^{-1}\|v_i\|_2\|v_j\|_2. \]
\end{lemma}
\begin{proof}
By scaling we may assume that $\|v_i\|_2 = 1$ for all $i$. Let $\mV \in \bbC^{q \times d}$ be the matrix whose $i$-th row is $v_i$. Then $\mV\mV^*$ has entries $\l v_i, v_j \r$, and has rank at most $d$. Hence
\[ \sum_{i,j} |\l v_i, v_j \r|^2 = \mathsf{Tr}((\mV\mV^*)^2) \ge d \cdot (q/d)^2 = q^2/d. \]
Thus there is an $i \neq j$ with
\[ |\l v_i, v_j \r|^2 \ge \frac{q^2/d - q}{q(q-1)} = \frac{q-d}{d(q-1)} \ge \frac{1}{d^2} \]
as long as $q \ge d+1$.
\end{proof}
Let $(s_1,t_1) \neq (s_2,t_2)$ (we do allow for $s_1 = s_2$ or $t_1 = t_2$) satisfy
\begin{align} |\l V^{(s_1,t_1)}, V^{(s_2,t_2)} \r| \ge |R|^{-1} \|V^{(s_1,t_1)}\|_2 \|V^{(s_2,t_2)}\|_2 \ge \frac{1}{2M} \gamma_{n-1,d-1}^2, \label{eq:choice} \end{align}
where we have used \cref{lemma:bound1}. 

Now we define the perturbed functions $\wt{G}, \wt{H}$. Let $\alpha_s$ for $s \in S$ and $\beta_t$ for $t \in T$ be unit complex variables. Let $A_s$ for $s \in S$ and $B_t$ for $t \in T$ be nonnegative real numbers satisfying $\sum_{s\in S} A_s^2 = 1$ and $\sum_{t \in T} B_t^2 = 1$. Define
\[ \wt{G} = \sum_{s \in S} \alpha_s A_s G_s \enspace \text{ and } \enspace \wt{H} = \sum_{t \in T} \beta_t B_t H_t. \]
Note that $\|\wt{G}\|_2 = \|\wt{H}\|_2 = 1$. Define the coefficients \[ \wt{F_r} = \E_{(x,y,z)\sim\mu^I}[F_r(x)\wt{G}(y)\wt{H}(z)\prod_{i\in I} T(x_i,y_i,z_i)], \] and the function $\wt{F} = \frac{\sum_{r\in R} \bar{\wt{F_r}}F_r}{\sqrt{\sum_{r\in R} |\wt{F_r}|^2}}.$ Then $\|\wt{F}\|_2 = 1$ and thus
\begin{align} \sum_{r \in R} |\wt{F_r}|^2 = \left|\E_{(x,y,z)\sim\mu^I}[\wt{F}(x)\wt{G}(y)\wt{H}(z)\prod_{i\in I} T(x_i,y_i,z_i)] \right|^2 \le \gamma_{n-1,d-1}^2, \label{eq:inductfr} \end{align}
for any choice of $\alpha_s, \beta_t$, $A_s$, $B_t$ by the inductive hypothesis. Expanding the leftmost term gives
\begin{align*} \sum_{r \in R} |\wt{F_r}|^2 &= \sum_{r,s,t,s',t'} \alpha_s A_s \bar{\alpha_{s'}} A_{s'} \beta_t B_t \bar{\beta_{t'}} B_{t'} \hat{F_r}(s,t) \bar{\hat{F_r}(s',t')} \\
&= \sum_{s,t,s',t'} \alpha_s A_s \bar{\alpha_{s'}} A_{s'} \beta_t B_t \bar{\beta_{t'}} B_{t'} \l V^{(s,t)}, V^{(s',t')}\r. \end{align*}
Let us take the expectation of the previous expression over $\alpha_s$ for $s \notin \{s_1, s_2\}$ and $\beta_t$ for $t \notin \{t_1, t_2\}$ being independent, uniform unit complex variables. This gives
\begin{align}
    &\mathop{\E}_{\substack{\alpha_s, s \notin \{s_1, s_2\} \\ \beta_t, t \notin \{t_1, t_2\}}}\left[\sum_{s,t,s',t'} \alpha_s \bar{\alpha_{s'}} \beta_t \bar{\beta_{t'}} \l V^{(s,t)}, V^{(s',t')}\r \right] \nonumber \\
    = ~& \sum_{s \in S,t \in T} A_s^2 B_t^2 \|V^{(s,t)}\|_2^2 \label{eq:11} \\
    + ~& \sum_{\substack{s \in S \\ t, t' \in \{t_1,t_2\}, t \neq t'}} A_s^2 \beta_t B_t \bar{\beta_{t'}} B_{t'} \l V^{(s,t)}, V^{(s,t')} \r \label{eq:12} \\
    + ~& \sum_{\substack{s, s' \in \{s_1, s_2\}, s \neq s' \\ t \in T}} \alpha_s A_s \bar{\alpha_{s'}} A_{s'} B_t^2 \l V^{(s,t)}, V^{(s',t)} \r \label{eq:21} \\
    + ~& \sum_{\substack{s, s' \in \{s_1, s_2\}, s \neq s' \\ t, t' \in \{t_1,t_2\}, t \neq t'}} \alpha_s A_s \bar{\alpha_{s'}} A_{s'} \beta_t B_t \bar{\beta_{t'}}  B_{t'} \l V^{(s,t)}, V^{(s',t')} \r. \label{eq:22}
    \end{align}
By \cref{cor:bound1} we know that the expression in \eqref{eq:11} is at least $(1 - \frac{4\delta M^3}{\Delta^4})\gamma_{n-1,d-1}^2$. The rest of the argument is split into subcases depending on whether $s_1 = s_2$, $t_1 = t_2$, or neither.
\paragraph{Case 1: We have $|\l V^{(s_1,t_1)}, V^{(s_2,t_2)}\r| > \frac{4\delta M^3}{\Delta^4}\gamma_{n-1,d-1}^2$ for some $s_1 = s_2$ or $t_1 = t_2$:} By symmetry, we only consider the case $s_1 = s_2$. Let $s = s_1$. In this case, the terms in \eqref{eq:21} and \eqref{eq:22} vanish.
We set $A_s = 1$ and all other $A_{s'} = 0$. set $B_{t_1} = 1/\sqrt{2}$ and $B_{t_2} = 1/\sqrt{2}$. The the expression in \eqref{eq:12} simplifies to
\[ \frac{1}{2}\beta_{t_1}\bar{\beta_{t_2}}\l V^{(s,t_1)}, V^{(s,t_2)}\r + \frac{1}{2}\beta_{t_2}\bar{\beta_{t_1}}\l V^{(s,t_2)}, V^{(s,t_1)}\r = \mathsf{Re}\left(\beta_{t_1}\bar{\beta_{t_2}}\l V^{(s,t_1)}, V^{(s,t_2)}\r \right). \]
Thus, we should choose $\beta_{t_1}\bar{\beta_{t_2}}$ to make the final expression a positive real number. Then its value is $|\l V^{(s,t_1)}, V^{(s,t_2)}\r| > \frac{4\delta M^3}{\Delta^4}\gamma_{n-1,d-1}^2$ by assumption. So the total sum is greater than
\[ \left(1 - \frac{4\delta M^3}{\Delta^4}\right)\gamma_{n-1,d-1}^2 + \frac{4\delta M^3}{\Delta^4}\gamma_{n-1,d-1}^2 = \gamma_{n-1,d-1}^2, \] a contradiction.

\paragraph{Case 1 does not hold and $|\l V^{(s_1,t_1)}, V^{(s_2,t_2)}\r| \ge \frac{1}{2M}\gamma_{n-1,d-1}^2$ for some $s_1 \neq s_2$ and $t_1 \neq t_2$:} 
Set $A_{s_1} = A_{s_2} = B_{t_1} = B_{t_2} = 1/\sqrt{2}$. We first bound the expressions in \eqref{eq:12} and \eqref{eq:21}. Because \textbf{Case 1} does not hold, the expression in \eqref{eq:12} can be lower bounded by
\[ -\max_{s \in \{s_1,s_2\}} |\l V^{(s,t_1)}, V^{(s,t_2)}| \ge -\frac{4\delta M^3}{\Delta^4}\gamma_{n-1,d-1}^2, \] and an identical bound follows for \eqref{eq:21}. The expression in \eqref{eq:12} can be simplified to
\begin{align}
\frac{1}{4}&\Bigg(\alpha_{s_1}\bar{\alpha_{s_2}} \beta_{t_1}\bar{\beta_{t_2}} \l V^{(s_1,t_1)}, V^{(s_2,t_2)}\r + \alpha_{s_2}\bar{\alpha_{s_1}} \beta_{t_2}\bar{\beta_{t_1}} \l V^{(s_2,t_2)}, V^{(s_1,t_1)}\r \nonumber \\ +~&\alpha_{s_2}\bar{\alpha_{s_1}} \beta_{t_1}\bar{\beta_{t_2}} \l V^{(s_2,t_1)}, V^{(s_1,t_2)}\r + \alpha_{s_1}\bar{\alpha_{s_2}} \beta_{t_2}\bar{\beta_{t_1}} \l V^{(s_1,t_2)}, V^{(s_2,t_1)}\r \Bigg) \nonumber \\
= \frac{1}{2}~&\Re\left(\alpha_{s_1}\bar{\alpha_{s_2}} \beta_{t_1}\bar{\beta_{t_2}} \l V^{(s_1,t_1)}, V^{(s_2,t_2)}\r + \alpha_{s_2}\bar{\alpha_{s_1}} \beta_{t_1}\bar{\beta_{t_2}} \l V^{(s_2,t_1)}, V^{(s_1,t_2)}\r \right). \label{eq:bigstuff}
\end{align}
Define $\alpha = \alpha_{s_1}\bar{\alpha_{s_2}}$ and $\beta = \beta_{t_1} \bar{\beta_{t_2}}$. Pick $\alpha,\beta$ uniformly at random so that $\alpha \beta \l V^{(s_1,t_1)}, V^{(s_2,t_2)} \r$ is a positive real number. Say $\alpha\beta = \omega$. Then $\E[\alpha^{-1}\beta] = \E[\omega^{-1}\beta^2] = 0$, so the second term in \eqref{eq:bigstuff} vanishes.
Thus in expectation, the total value of equations \eqref{eq:11} to \eqref{eq:22} is at least
\[ \left(1 - \frac{4\delta M^3}{\Delta^4}\right)\gamma_{n-1,d-1}^2 - \frac{4\delta M^3}{\Delta^4}\gamma_{n-1,d-1}^2 - \frac{4\delta M^3}{\Delta^4}\gamma_{n-1,d-1}^2 + \frac{1}{2M}\gamma_{n-1,d-1}^2 > \gamma_{n-1,d-1}^2, \] a contradiction.

\subsection{Handling the Case \texorpdfstring{$|R| \ge |S||T|$}{r>st}}
\label{subsec:r>st}

In this case we use the relaxed base case (\cref{def:relaxed}). We start with the following observation.
\begin{lemma}
\label{lemma:fr'large}
If $|R| \ge |S||T|$ then for all $s \in S, t \in T$ we have $\sum_{r \in R} |\hat{F_r'}(s,t)|^2 \ge 1 - \frac{4\delta M^5}{\Delta^2}$.
\end{lemma}
\begin{proof}
By \cref{cor:bound2} we know that
\[ \sum_{r \in R, s \in S, t \in T} |\hat{F_r'}(s,t)|^2 \ge |R|\left(1 - \frac{4\delta M^3}{\Delta^2}\right) \ge |S||T| \left(1 - \frac{4\delta M^3}{\Delta^2}\right). \]
Because $\sum_{r \in R} |\hat{F_r'}(s,t)|^2 \le 1$ for all $s \in S, t \in T$ by \cref{lemma:bound21}, the conclusion follows.
\end{proof}
For $s \in S, t \in T$, define the functions
\[ \wt{F_{st}} = \frac{\sum_{r \in R} \bar{\hat{F_r'}(s,t)}F_r'}{\sqrt{\sum_{r\in R} |\hat{F_r'}(s,t)|^2}}. \]
Note that $\|\wt{F_{st}}\|_2 = 1$ and
\begin{align} \left(\E_{(x,y,z)\sim\mu}\left[\wt{F_{st}}(x) G_s'(y) H_t'(z) T(x,y,z) \right]\right)^2 = \sum_{r \in R} |\hat{F_r'}(s,t)|^2 \ge 1 - \frac{4\delta M^5}{\Delta^2}, \label{eq:fr'2} \end{align}
by \cref{lemma:fr'large}. Now we leverage the relaxed base case property of $\mu$ to get that \[ I_{\eff}[\wt{F_{st}}] \le \left(\frac{4\delta M^5}{c\Delta^2}\right)^{\frac{1}{100M}} := \delta', \] for all $s \in S$ and $t \in T$. The remainder of the proof converts this into an effective influence bound on the $F_r'$ functions themselves to contradict \cref{lemma:variance}.

In the following, treat all functions to be on the space $L_2(\B \times \C, \mu_{y,z})$, where we recall that $(y,z)$ is a product distribution and determines $x$ uniquely in the relaxed base case. Note that for all $s \in S, t \in T$ that
\begin{align} \|\bar{\wt{F_{st}}} - G_s' H_t' T\|_2 = \|\wt{F_{st}}\|_2^2 + \|G_s' H_t' T\|_2^2 - 2 \E_{(x,y,z)}\left[\wt{F_{st}}(x) G_s'(y) H_t'(z) T(x,y,z) \right] \le \frac{8\delta M^5}{\Delta^2}, \label{eq:kobe} \end{align} where we have used \eqref{eq:fr'2} and that the expectation is real. For all $r \in R$ we have
\[ \Big\|\bar{F_r'} - \sum_{s\in S, t\in T} \bar{\hat{F_r'}(s,t)} \bar{\wt{F_{st}}} \Big\|_2 \le \Big\|\bar{F_r'} - \sum_{s \in S, t \in T} \bar{\hat{F_r'}(s,t)} G_s' H_t' T \Big\|_2 + \sum_{s \in S, t \in T} \hat{F_r'}(s,t)\left\|G_s' H_t' T - \bar{\wt{F_{st}}} \right\|_2. \]
By \eqref{eq:kobe} the second term can be bounded by
\[ \sum_{s \in S, t \in T} \hat{F_r'}(s,t)\left\|G_s' H_t' T - \bar{\wt{F_{st}}} \right\|_2 \le M^2 \cdot \frac{8\delta M^5}{\Delta^2} = \frac{8\delta M^7}{\Delta^2}. \]
The first term can be simplified to get
\[ \Big\|\bar{F_r'} - \sum_{s \in S, t \in T} \bar{\hat{F_r'}(s,t)} G_s' H_t' T \Big\|_2 = \Big(\|F_r'\|_2^2 - \sum_{s \in S, t \in T} |\hat{F_r'}(s,t)|^2 \Big)^{1/2} \le \sqrt{\frac{4\delta M^3}{\Delta^2}},\]
where we have used that the functions $G_s'H_t'T$ are orthonormal (see the proof of \cref{lemma:bound2}), and \cref{cor:bound2}. By the triangle inequality again, we get
\begin{align*} I_{\eff}[F_r'] &\le I_{\eff}\Big[\bar{F_r'} - \sum_{s\in S, t\in T} \bar{\hat{F_r'}(s,t)} \bar{\wt{F_{st}}}\Big] + \sum_{s \in S, t \in T} |\hat{F_r'}(s,t)| I_{\eff}[\wt{F_{st}}]
\\ &\le \Big\|\bar{F_r'} - \sum_{s\in S, t\in T} \bar{\hat{F_r'}(s,t)} \bar{\wt{F_{st}}}\Big\|_2 + \sum_{s \in S, t \in T} |\hat{F_r'}(s,t)| I_{\eff}[\wt{F_{st}}] \\ 
&\le \sqrt{\frac{4\delta M^3}{\Delta^2}} + \frac{8\delta M^7}{\Delta^2} + M^2 \delta' \le 2M^2 \delta',
\end{align*}
by our choice of $\delta = \left(\frac{\Delta}{CM\tau}\right)^{1000M}.$ This gives that
\[ \sum_{r \in [M]} \aa_r^2 I_{\eff}[F_r'] \le M \delta^2 + \sum_{r \in R} \aa_r^2 I_{\eff}[F_r'] \le M\delta^2 + 2M^2 \delta' < \tau, \]
by the choice of $\delta$ again, because $\delta' = \left(\frac{4\delta M^5}{c\Delta^2}\right)^{\frac{1}{100M}}$. This contradicts \cref{lemma:variance}.

\section{Complex Embeddings: Reductions}
\label{sec:complex}

The system $f(x) + g(y) + h(z) = t(x, y, z) \pmod{m}$ has a solution for all $(x,y,z) \in \supp(\mu)$ iff it has a solution for all prime power factors $p^k$ of $m$ by the Chinese remainder theorem. So we assume $m = p^k$ for the remainder of the paper. If $T$ is nonembeddable over $\D$, then we are done by \cref{cor:complex}. So we assume that $T$ is embeddable over $\D$, and hence there is an integer $N$ and functions $a: \A \to \Z_{p^k N}, b: \B \to \Z_{p^k N}, c: \C \to \Z_{p^k N}$ with
\[ a(x) + b(y) + c(z) = N \cdot t(x,y,z) \pmod{p^k N}. \]
We assume that $N$ is minimal.
Thus the equation $f(x) + g(y) + h(z) = t(x, y, z) \pmod{p^k}$ can be rewritten as $\wt{f}(x) + \wt{g}(y) + \wt{h}(z) = 0 \pmod{p^k N}$ where $\wt{f}(x) = N \cdot f(x) - a(x)$.

This motivates us to apply the techniques of \cite{BKM4} to reduce to the case where the distribution $\mu$ is over all $(x,y,z) \in \cA^3$ with $x+y+z=0$, where $\cA$ is a finite abelian group (the \emph{master group}) that governs all abelian embeddings of $\mu$.
\subsection{Master Group and Embedding}
\label{subsec:master}
We start by introducing the master embedding, which essentially captures all the abelian embeddings of $\mu$. The master embedding maps into a finite abelian group as long as $\mu$ has no $\Z$-embeddings.

\begin{definition}[Master Embedding]
\label{def:master}
Let $\mu$ be a distribution on $\A \times \B \times \C$, and let $M = \max\{|\A|, |\B|, |\C|\}$. Let $r = O_M(1)$ be sufficiently large. Let $(\a_1, \b_1, \c_1), \dots, (\a_s, \b_s, \c_s)$ be all embeddings of $\mu$ into cyclic groups $\cA_1, \dots, \cA_s$ of size at most $r$. Then a master embedding of $\mu$ into $\prod_{i\in[s]} \cA_i$ is given by
\[ \ama(x) = (\a_1(x), \dots, \a_s(x)), \enspace \bma(y) = (\b_1(y), \dots, \b_s(y)), \enspace \cma(z) = (\c_1(z), \dots, \c_s(z)). \]
\end{definition}
As argued in \cite{BKM4}, for sufficiently large $r = O_M(1)$, the master embedding captures all embeddings that are not equal up to simple transformations.
\begin{lemma}
\label{lemma:same}
If $\mu$ has no $\Z$-embeddings, then $\a(x) = \a(x')$ for all embeddings $\a, \b, \c$ of $\mu$ if and only if $\ama(x) = \ama(x')$, and simliar for the $y, z$ variables.
\end{lemma}
We also do a reduction to the embedding functions $a(x) + b(y) + c(z) = N \cdot t(x, y, z) \pmod{p^k N}$ to ensure that $(a, b, c)$ are constant on the master embedding. Specifically, define $\wt{a}(x) = a(x) \pmod{N}$, and similar for $\wt{b}, \wt{c}$, and let $\wt{t}(x, y, z) = (\wt{a}(x) + \wt{b}(y) + \wt{c}(z))/N$. We claim that the game with target values $\wt{t}$ is equivalent to the game with target value $t$. Indeed, note that if $f(x) + g(y) + h(z) = t(x, y, z)$ then for
\[ \wt{f}(x) = f(x) + \frac{\wt{a}(x) - a(x)}{N}, \enspace \wt{g}(y) = g(y) + \frac{\wt{b}(y) - b(y)}{N}, \enspace \wt{h}(z) = h(z) + \frac{\wt{c}(z) - c(z)}{N} \] we have $\wt{f}(x) + \wt{g}(y) + \wt{h}(z) = \wt{t}(x, y, z) \pmod{p^k}$.

Thus, from now on, we make the following assumption.
\begin{assumption}
\label{ass:bounded}
We write $N \cdot t(x, y, z) = a(x) + b(y) + c(z) \pmod{p^k N}$ where $0 \le a(x), b(y), c(z) < N$ for all $x, y, z$.
\end{assumption}
This way, we know that $a, b, c$ is an embedding, hence it is captured by the master embedding.
\begin{lemma}
\label{lemma:ass}
Under \cref{ass:bounded}, $(a, b, c)$ is an embedding. Thus, if $\ama(x) = \ama(x')$, we have $a(x) = a(x')$.
\end{lemma}

\subsection{Modest Noise Operator and Modest Influences}
\label{subsec:modest}
The goal of this section is to define the \emph{modest noise operator} and \emph{modest influences}.

Let $\mu$ be a distribution on $\A \times \B \times \C$, and consider a set $\A' \subseteq \A$. Let $\cA$ be the master group, and for each $a \in \cA$, let $\A'_a = \{x \in \A' : \ama(x) = a\}.$ Create an orthonormal basis of $L^2(\A, \mu_x)$ called $B = B' \cup \cup_{a \in \cA} B_a$ as follows. $B_a$ is an orthnormal basis for functions supported on $\A'_a$ that are orthogonal to constant. $B'$ is an orthnormal basis for functions that are constant on each $\A'_a$.

Any function $F \in L_2(\A^n, \mu_x^{\otimes n})$ can be written as $F = \sum_{\chi \in B^{\otimes n}} \hat{F}(\chi) \chi$. For a character $\chi$, define its \emph{modest degree} as $\moddeg(\chi) = \{i : \chi_i \in \cup_{a \in \cA} B_a\}$. This allows us to define the modest degree decomposition
\[ F^{\modest=i} = \sum_{\chi \in B^{\otimes n} : \moddeg(\chi) = i} \hat{F}(\chi)\chi. \]
Now define the modest influence as
\[ I_{\modest}[F] = \sum_{i\in[n]} I_{i,\modest}[F] = \E_{(x,y)\in\A^n}[|F(x) - F(y)|^2], \]
where $(x,y)$ are sampled as follows. Sample $x \in \A^n$ according to $\mu_x^{\otimes n}$. If $x_i \notin \A'$ set $y_i = x_i$. Otherwise, set $y_i$ to be from $\mu_x$ conditioned on $y_i \in \A'_{\ama(x_i)}$.

Analogous to the case of standard influences we know that
\[ I_{i,\modest}[F] = \sum_{\chi \in B^{\otimes n}, \chi_i \in B \setminus B'} |\hat{F}(\chi)|^2 \enspace \text{ and } \enspace I_{\modest}[F] = \sum_{\chi \in B^{\otimes n}} \moddeg(\chi) |\hat{F}(\chi)|^2. \]

We finally define the modest noise operator $\TT^{\modest}_{1-\rho}$. Define $\TT^{\modest}_{1-\rho}F(x) = \E_{y \sim N^{\modest}_{1-\rho}(x)}[F(y)]$, where $N^{\modest}_{1-\rho}(x)$ is the following distribution. For each $i \in [n]$, with probability $1-\rho$, set $y_i = x_i.$ Otherwise with probability $\rho$, do the following. If $x_i \in \A \setminus \A'$, set $y_i = x_i$. Otherwise, sample $y_i$ from $\mu_x$ conditioned on $y_i \in B'_{\ama(x_i)}$. Analogous to the standard noise operator, we have
\[ \TT^{\modest}_{1-\rho}F = \sum_{0 \le i \le n} (1-\rho)^i F^{\modest=i}. \]

\subsection{Saturating the Master Embedding via Path Tricks}
\label{subsec:saturate}
The goal of this section is to perform multiple path tricks as in \cref{subsec:pathtrick}. This is with two goals in mind. The first is to enlarge the sets $\{\ama(x) : x \in \A\}, \{\bma(y) : y \in \B\}, \{\cma(z) : z \in \C\}$ until each of them forms a group. In fact, it is easy to argue that the three groups must be identical.

The secondary goal will be to work towards establishing a variation of the relaxed base case.

Throughout this section, $\mu, t$ will denote a distribution and target value, and $f: \A^n \to \Z_m^n, g: \B^n \to \Z_m^n, h: \C^n \to \Z_m^n$ will be strategies. Let the winning probability of $f, g, h$ be represented as
\[ \val_{\mu,t}(f,g,h) = \Pr_{(x,y,z)\sim\mu^{\otimes n}}\left[f(x)_i + g(y)_i + h(z)_i = t(x_i,y_i,z_i)\pmod{m} \enspace \forall \enspace i = 1, 2, \dots, n \right]. \]
\begin{lemma}
\label{lemma:npathtrick}
Let $\mu, t$ represent a game of value less than $1$, and let $\val_{\mu,t}(f,g,h) \ge \eps$. If $\mu^+$ represents the result of applying a path trick to $\mu$ with $r$ steps (\cref{def:pathtrick}), then there is $t^+$ such that $\mu^+, t^+$ has value less than $1$, and there are $\wt{f}, \wt{g}, \wt{h}$ with $\val_{\mu^+,t^+}(\wt{f}, \wt{g}, \wt{h}) \ge \eps^{2^r}$.
\end{lemma}
\begin{proof}
Without loss of generality, let $\mu^+$ be the result of an $x$ path trick. In this case, we set
\[ t^+(\vec{x}, y_1, z_{2^{r-1}-1}) = \sum_{i=1}^{2^{r-1}} t(x_i,y_i,z_i) - \sum_{i=1}^{2^{r-1}-1} t(x_i',y_{i+1},z_i). \]
This definition actually does not depend on $y_2, \dots, y_{2^{r-1}}, z_1, \dots, z_{2^{r-1}-2}$. Indeed, because
$a(x) + b(y) + c(z) = N \cdot t(x,y,z) \pmod{p^k N}$ we have \[ N \cdot t^+(\vec{x}, y_1, z_{2^{r-1}-1}) = \sum_{i=1}^{2^{r-1}} a(x_i) - \sum_{i=1}^{2^{r-1}-1} a(x_i') + b(y_1) + c(z_{2^{r-1}-1}) \pmod{p^k N}, \] so $t^+$ is uniquely determined modulo $p^k$.

Now, $\mu^+, t^+$ has value less than $1$ because for all $(x,y,z) \in \supp(\mu)$, we have $t^+((x,\dots,x), y, z) = t(x,y,z)$. We will show the following by induction.
\begin{align} \mathop{\Pr}_{\substack{x_1,x_1',\dots,x_{2^{r-1}}, x_{2^{r-1}}' \\ y_1, \dots, y_{2^{r-1}+1} \\ z_1, \dots, z_{2^{r-1}}}} \left[\sum_{i=1}^{2^{r-1}} f(x_i)-f(x_i') + g(y_1) - g(y^{2^{r-1}+1}) = \sum_{i=1}^{2^{r-1}} t(x_i,y_i,z_i) - t(x_i',y_{i+1},z_i) \right] \ge \eps^{2^r}. \label{eq:maininduct} \end{align}
Let us establish the case $r=1$. Because $\val_{\mu,t}(f,g,h) \ge \eps$, we get
\begin{align}
    \eps^2 &\le \E_z\left[\Pr_{(x_1,y_1,z_1) \sim \mu)}[f(x_1)+g(y_1)-t(x_1,y_1,z_1) = h(z_1) | z_1 = z]\right]^2 \nonumber \\
    &\le \E_{z'}\left[\Pr_{(x_1,y_1,z_1) \sim \mu)}[f(x_1)+g(y_1)-t(x_1,y_1,z_1) = h(z_1) | z_1 = z]^2 \right] \nonumber \\
    &= \E_z\Pr_{\substack{(x_1,y_1,z_1) \sim \mu \\ (x_1',y_2,z_1) \sim \mu}}[f(x_1)+g(y_1)-t(x_1,y_1,z_1) = f(x_2) + g(y_2) - t(x_2,y_2,z_1) = h(z_1) | z_1 = z]  \nonumber \\
    &\le \E_z\Pr_{\substack{(x_1,y_1,z_1) \sim \mu \\ (x_1',y_2,z_1) \sim \mu}}[f(x_1)-f(x_2)+g(y_1)-g(y_2) = t(x_1,y_1,z_1) - t(x_1',y_2,z_1) | z_1 = z] \label{eq:stuff1} \\
    &= \mathop{\Pr}_{\substack{x_1,x_1' \\ y_1, y_2 \\ z_1}}[f(x_1)-f(x_1')+g(y_1)-g(y_2) = t(x_1,y_1,z_1) - t(x_1',y_2,z_1)]. \label{eq:stuff2}
\end{align}
We complete the induction with a similar approach. By induction
\begin{align*}
    \eps^{2^{r+1}} &\le \E_{y_{2^{r-1}+1}}\left[\mathop{\Pr}_{\substack{x_1,x_1',\dots,x_{2^{r-1}}, x_{2^{r-1}}' \\ y_1, \dots, y_{2^{r-1}} \\ z_1, \dots, z_{2^{r-1}}}} \left[\sum_{i=1}^{2^{r-1}} f(x_i)-f(x_i') + g(y_1) - g(y^{2^{r-1}+1}) = \sum_{i=1}^{2^{r-1}} t(x_i,y_i,z_i) - t(x_i',y_{i+1},z_i) \right] \right]^2 \\
    &\le \E_{y_{2^{r-1}+1}}\left[\mathop{\Pr}_{\substack{x_1,x_1',\dots,x_{2^{r-1}}, x_{2^{r-1}}' \\ y_1, \dots, y_{2^{r-1}} \\ z_1, \dots, z_{2^{r-1}}}} \left[\sum_{i=1}^{2^{r-1}} f(x_i)-f(x_i') + g(y_1) - g(y^{2^{r-1}+1}) = \sum_{i=1}^{2^{r-1}} t(x_i,y_i,z_i) - t(x_i',y_{i+1},z_i) \right]^2 \right].
\end{align*}
Expanding the last expression and removing the $g(y^{2^{r-1}+1})$ like in \eqref{eq:stuff1}, \eqref{eq:stuff2} completes the induction.
To conclude, we discuss the choices of $\wt{f}, \wt{g}, \wt{h}$. We set $\wt{f}(\vec{x}) = \sum_{i=1}^{2^{r-1}} f(x_i) - \sum_{i=1}^{2^{r-1}-1} f(x_i')$ and $\wt{g}(y) = y$. Finally, for each $z \in \C$ we randomly choose $\wt{h}(z)$ as follows: generate $(x,y,z')$ from $\mu$ conditioned on $z' = z$, and set $\wt{h}(z) = t(x,y,z') - f(x) - g(y)$. Then \eqref{eq:maininduct} exactly tells us that $\E_{\wt{h}}[\val_{\mu^+,t^+}(\wt{f}, \wt{g}, \wt{h})] \ge \eps^{2^r}$ as desired.
\end{proof}
We study how embeddings evolve under the path trick.
\begin{lemma}
\label{lemma:pathtrickembed}
Let $\mu^+$ be the result of applying an $x$ path trick to $\mu$. If $\a^+, \b^+, \c^+$ is an embedding of $\mu^+$, then there must be an embedding $\a, \b, \c$ of $\mu$ such that $\b^+(y) = b(y)$ for all $y \in \B$, $\c^+(z) = \c(z)$ for all $z \in \C$, and $\a^+(\vec{x}) = \sum_{i=1}^{2^{r-1}} \a(x_i) - \sum_{i=1}^{2^{r-1}-1} \a(x_i')$.
\end{lemma}
\begin{proof}
Recall that $((x,\dots,x), y, z) \in \supp(\mu^+)$. Thus, there is are embeddings $\a, \b, \c$ of $\mu$ with $\b^+(y) = \b(y)$ and $\c^+(z) = \c(z)$. If $(\vec{x}, y_1, z_{2^{r-1}-1}) \in \supp(\mu^+)$, then 
\[ \a^+(\vec{x}) = -\b(y_1) - \c(z_{2^{r-1}-1}) = \sum_{i=1}^{2^{r-1}} \a(x_i) - \sum_{i=1}^{2^{r-1}-1} \a(x_i'). \]
\end{proof}
Finally, we discuss how to augment the image of the master embedding.
\begin{lemma}
\label{lemma:augment}
Consider the following sequence of steps to a distribution $\mu$.
\begin{enumerate}
    \item Perform a $z$ path trick with $r > \log_2 M$.
    \item Perform a $y$ path trick with $r > \log_2 M$.
    \item Perform a $x$ path trick with $r = 2$.
\end{enumerate}
Then if $\mu$ had master embeddings $\ama, \bma, \cma$, then the new distribution $\mu^+$ has master embeddings $\ama^+, \bma^+, \cma^+$ generated according to \cref{lemma:pathtrickembed}. In particular, $\Image(\bma) \subseteq \Image(\bma^+)$, $\Image(\cma) \subseteq \Image(\cma^+)$, and for all $x_1, x_2, x_3 \in \A$ we have $\ama(x_1) - \ama(x_2) + \ama(x_3) \in \Image(\ama^+)$.
\end{lemma}
\begin{proof}
Everything before the final sentence follows from \cref{lemma:pathtrickembed}. Let $\mu', \mu''$ after step 1, 2 respectively. In $\mu'$, $(x,y)$ has full support. In $\mu''$, $(x,z)$ has full support, and for all $x \in \A, y \in \B$ there is some $z$ with $(x, (y, \dots, y), z) \in \supp(\mu'')$. Now, we find a path containing any $x_1, x_2, x_3 \in \A$.
Pick a $y \in \B$. Let $z_1, z_2$ be such that $(x_1, (y, \dots, y), z_1) \in \supp(\mu'')$ and $(x_2, (y, \dots, y), z_2) \in \supp(\mu'')$. Because $(x, z)$ has full support in $\mu''$, there is $\vec{y}$ with $(x_3, \vec{y}, z_2) \in \supp(\mu'')$, as desired.
\end{proof}
By shifting, assume that $0 \in \Image(\ama), \Image(\bma), \Image(\cma)$.
\begin{lemma}[Saturated embedding]
\label{lemma:saturated}
Let $\mu, t$ represent a game of value less than $1$, with strategies $f, g, h$ with $\val_{\mu,t}(f,g,h) \ge \eps$. Then there is a game $\mu^+, t^+$ of value less than $1$, and strategies $\wt{f}, \wt{g}, \wt{h}$ with $\val_{\mu^+, t^+}(\wt{f}, \wt{g}, \wt{h}) \ge \eps^{O_M(1)}$.
Additionally:
\begin{enumerate}
    \item The master embeddings $\ama, \bma, \cma$ of $\mu^+$ satisfy for some finite group $\cA$: \[ \Image(\ama) = \Image(\bma) = \Image(\cma) = \cA. \]
    \item $(y, z)$ has full support in $\mu^+$.
\end{enumerate}
\end{lemma}
\begin{proof}
Repeating \cref{lemma:augment} multiple times, and its symmetric $y, z$ versions, we eventually reach a distribution $\mu^+$ where $\Image(\ama)$, $\Image(\bma)$, $\Image(\cma)$ are subgroups. These must be the same subgroup because (say) $(x,y)$ has full support. Then for any $a \in \Image(\ama), b \in \Image(\bma)$, we know that $-a-b \in \Image(\cma)$. Let $\cA = \Image(\ama)$. Item 2 can be guaranteed by applying one more $x$ path trick. The bound $\val_{\mu^+,t^+}(\wt{f}, \wt{g}, \wt{h}) \ge \eps^{O_M(1)}$ follows since we apply \cref{lemma:npathtrick} a finite number of times.
\end{proof}

\subsection{Merging Symbols}
\label{subsec:merginggame}
We apply symbol merging in this section as in \cref{subsec:merging}. 
\begin{lemma}
\label{lemma:merginggame}
Let $(\mu, t)$ be a game of value less than $1$, with strategies $f, g, h$ with $\val_{\mu,t}(f,g,h) \ge \eps$. Then there is a distribution $\mu'$ with $\{(\rep(x), y, z) : (x,y,z) \in \supp(\mu)\} \subseteq \supp(\mu')$, $t'$ on $\supp(\mu')$ with game value less than $1$, and strategy $\wt{f}$ such that $\val_{\mu',t'}(\wt{f},g,h) \ge \eps^{2M}$.
\end{lemma}
\begin{proof}
To avoid repeating arguments, we will show this by combining \cref{lemma:identity,lemma:merging}. We use the same notation as the proof of \cref{lemma:merging}.
In particular, $\mu'$ is still given by the random walk between $\A$ and $(\B \times \C)$, and we let $N \cdot t'(x,y,z) = a(x) + b(y) + c(z) \pmod{p^k N}$. Clearly, $(\mu',t')$ has game value less than $1$ because $t'(x,y,z) = t(x,y,z)$ for $(x,y,z) \in \supp(\mu)$.
Finally, we set $\wt{f}(x)$ as follows: pick $(y,z)$ as a random neighbor of $x$ propotional to mass in $\mu$, and set $\wt{f}(x) = t(x,y,z) - g(y) - h(z)$.
Define $\wt{F}_S, G_S, H_S$ as in \cref{lemma:identity}, and $R_S(y,z) = G_S(y)H_s(z)$.
\begin{align*}
    \E_{\wt{f}}[\val_{\mu',t'}(\wt{f}, g, h)] &= \E_{S \in \Z_m^n}\E_{\wt{f}}[\wt{F}_S^* \mA(\mD_{yz}^{-1} \mA^* \mD_x^{-1} \mA)^{M-1} R_S]
    \\ &= \E_{S \in \Z_m^n}\left[(\mD_{yz}^{1/2}R_S)^* (\mD_{yz}^{-1/2} \mA^* \mD_x^{-1} \mA \mD_{yz}^{-1/2})^M \mD_{yz}^{1/2}R_S \right] \\
    &\ge \E_{S \in \Z_m^n}\left[(R_S^* \mA^* \mD_x^{-1} \mA R_S)^M \right] \ge \E_{S \in \Z_m^n}\left[R_S^* \mA^* \mD_x^{-1} \mA R_S \right]^M \\
    &\ge \E_{S \in \Z_m^n}[|\bar{F_S}^* \mA R_S|^2]^M \ge \E_{S \in \Z_m^n}[\bar{F_S}^* \mA R_S]^{2M} \\
    &= \val_{\mu,t}(f, g, h)^{2M} \ge \eps^{2M},
\end{align*}
where we have applied Jensen's inequality and \cref{lemma:identity} several times.
\end{proof}
Now we show that we can reduce to when the support is exactly $(\rep(x), y, z)$. This is essentially a repeat of the arguments in \cref{lemma:merging2,lemma:validt}.
\begin{lemma}
\label{lemma:merginggame2}
Let $(\mu, t)$ be a game of value less than $1$, with strategies $f, g, h$ with $\val_{\mu,t}(f,g,h) \ge \eps$. Then there is a distribution $\mu^-$ with support $\{(\rep(x), y, z) : (x,y,z) \in \supp(\mu)\}$, $(y,z)$ is a product distribution, $t^-$ on $\supp(\mu^-)$ with game value less than $1$, and strategies $\wt{f}, \wt{g}, \wt{h}$ on $n' \ge cn$ variables with $\val_{\mu^-,t^-}(\wt{f}, \wt{g}, \wt{h}) \ge \eps^{2M}/2$.
\end{lemma}
\begin{proof}
We may assume that $(\mu, t)$ are as constructed in \cref{lemma:merginggame}. Then $\mu^-, t^-, \wt{f}, \wt{g}, \wt{h}$ are formed by taking a random restriction. It suffices to show that $(\mu^-, t^-)$ has game value less than $1$. We show that otherwise, the two player game on $\A \times (\B \times \C)$ with target value $t(x,y,z)$ was unwinnable, so we can apply a two-player parallel repetition theorem to reach a contradiction \cite{Raz98}. Indeed, assume there are strategies $r(x) + s(y,z) = t(x,y,z) \pmod{m}$, and $f(x) + g(y) + h(z) = t^-(x,y,z)$ for $(x,y,z) \in \supp(\mu^-)$. Then clearly
\begin{align*} 
t(x,y,z) &= r(x) + s(y,z) = (r(x) - r(\rep(x)) + r(\rep(x)) + s(y,z) \\ &= (r(x) - r(\rep(x)) + t^-(x,y,z) = \left[f(x) + r(x) - r(\rep(x))\right] + g(y) + h(z),
\end{align*}
so $t$ admits a strategy with value $1$, a contradiction.
\end{proof}

\subsection{Establishing the Master Relaxed Base Case}
\label{subsec:masterrelaxed}
The goal of this section is to establish an analogue of the relaxed base case in this setting.
\begin{definition}[Master relaxed base case]
\label{def:masterrelaxed}
We say that a distribution $\mu$ on $\A \times \B \times \C$ and $\A' \subseteq \A$ satisfies the \emph{master relaxed base case} if:
\begin{itemize}
    \item $\mu_{y,z}$ is a product distribution.
    \item $(y, z)$ uniquely determines $x$ in $\supp(\mu)$.
    \item For functions $G: \B \to \bbC$, $H: \C \to \bbC$, and $F: \A \to \bbC$ with
    \[ I_{\modest}[F] = \E_{x,x'\sim\mu_x}[1_{x,x'\in\A'} 1_{\ama(x) = \ama(x')} |F(x) - F(x')|^2] \ge \tau\|F\|_2^2, \]
    we have for a constant $c := c(\mu)$ and $M = \max\{|\A|, |\B|, |\C|\}$ that \[ \E_{(x,y,z)\sim\mu}[F(x)G(y)H(z)] \le (1-c\tau^{100M})\|F\|_2 \|G\|_2 \|H\|_2. \]
\end{itemize}
\end{definition}
As in \cref{lemma:relaxed}, we show that applying two path tricks, merging, and taking a random restriction results in a distribution $\mu$ that satisfies the master relaxed base case. In the next section, it will be important to remember that we got to a distribution satisfying the relaxed base case in this manner, because we require that the original distribution is contained inside it.
\begin{lemma}
\label{lemma:masterrelaxed}
    Consider a distribution $\mu$ on $\A \times \B \times \C$. Perform the following sequence of operations to $\mu$.
    \begin{enumerate}
        \item Apply the path trick (\cref{lemma:npathtrick}) to $\mu$ on coordinate $y$ to reach a distribution $\mu^{(1)}$ on $\A \times \B^+ \times \C$.
        \item Apply the path trick (\cref{lemma:npathtrick}) to $\mu^{(1)}$ on coordinate $x$ to reach a distribution $\mu^{(2)}$ on $\A^+ \times \B^+ \times \C$.
        \item Merge symbols in $\A^+$ using (\cref{lemma:merginggame2}) to reach a distribution $\mu^{(3)}$.
        \item Random restrict $\mu^{(3)}$ so that $(y,z)$ is uniform. Call this distribution $\mu^+$.
        \item Set $\A' = \{\rep((x,\dots,x)) : x \in \A\}$.
    \end{enumerate}
Then $\mu^+$ satisfies the master relaxed base case with the set $\A'$.
\end{lemma}
\begin{proof}
The first two items of \cref{def:masterrelaxed} are satisfied by construction. We focus on the last item.

Relabel $\A^+, \B^+$ as $\A, \B$ for simplicity.
Start with the observation that $\E_{x\sim\mu_x}[1_{x\in\A'} |F(x)|^2] \ge \tau\|F\|_2^2$. Assume that $\|F\|_2 = \|G\|_2 = \|H\|_2 = 1$. For simplicity of notation, we use $\ls, \gs$ in this proof to suppress constants depending on $M, \mu$.

Let us multiply $F(x)$ pointwise by the same unit complex number so that 
\[ \E_{(x,y,z)\sim\mu}\left[F(x)G(y)H(z)\right] \] is a positive real number. We may assume this throughout. Then we have the equality
\[ \left|\E_{(x,y,z)\sim\mu}\left[F(x)G(y)H(z) \right]\right| = 1 - \frac12\E_{(x,y,z)\sim\mu}\left[|\bar{F(x)} - G(y)H(z)|^2 \right]. \]
Assume for contradiction that $|\bar{F(x)} - G(y)H(z)|^2 \le c\tau^{100M}$ for all  $(x,y,z) \in \supp(\mu)$. In particular, $||F(x)| - |G(y)||H(z)|| \le c\tau^{50M}$.

As in the proof of \cref{lemma:relaxed}, we can establish that $|F(x)|, |G(y)|, |H(z)| \gs \tau^{M/2}$ for all $x \in \A, y \in \B, z \in \C$. Now, we actually establish that $|F(x)|$ must be nearly constant by using that $\mu$ has no $\Z$-embeddings. Recall that $||F(x)| - |G(y)| |H(z)|| \ls \tau^{50M}$. Using that $|F(x)| \gs \tau^{M/2}$ gives us $|1 - |F(x)|^{-1} |G(y)| |H(z)|| \ls \tau^{49M}$. Define $f'(x) = -\log |F(x)|$, $g'(y) = \log |G(y)|$, $h'(z) = \log |H(z)|$. Then we get
\[ |f'(x) + g'(y) + h'(z)| \ls \tau^{49M} \enspace \text{ for all } \enspace (x, y, z) \in \supp(\mu), \] i.e., $f', g', h'$ are approximately a $\Z$-embedding. By the Dirichlet approximation theorem, for $N = (C\tau)^{9M}$, there is some $q \le N$ and integers $\a(x), \b(y), \c(z)$ with
\[ \left|f'(x) - \frac{\a(x)}{q} \right| \le \frac{1}{qN^{1/(3M)}}, \enspace \left|g'(y) - \frac{\b(y)}{q} \right| \le \frac{1}{qN^{1/(3M)}}, \enspace \left|h'(z) - \frac{\c(z)}{q} \right| \le \frac{1}{qN^{1/(3M)}}. \]
Thus
\begin{align*} |\a(x) + \b(y) + \c(z)|
\le q|f'(x) + g'(y) + h'(z)| + \frac{3}{N^{1/(3M)}} \ls N \tau^{49M} + 3(C\tau)^{-3} < 1.
\end{align*}
Hence $\a, \b, \c$ form a $\Z$-embedding, so they are all constant. Thus \[ |\log |F(x)| - \log |F'(x)|| \le \frac{1}{N^{1/(3M)}} = (C\tau)^{-3}, \] so $|1 - |F(x)|/|F(x')|| \ls (C\tau)^{-3}$ for all $x, x'$.

To conclude, we will establish that the arguments of $F(x)$ are nearly constant. Recall that $|\bar{F(x)} - G(y) H(z)| \ls \tau^{50M}$ and $|F(x)|, |G(y)|, |H(z)| \gs \tau^{M/2}$, so \[ \tau^{49M} \gs \left| \frac{\bar{F(x)}}{|F(x)|} - \frac{G(y)}{|G(y)|} \frac{H(z)}{|H(z)|} \right| = \left| 1 - \frac{F(x)}{|F(x)|} \frac{G(y)}{|G(y)|} \frac{H(z)}{|H(z)|} \right|. \]
Let $f'(x) = \arg(F(x))/(2\pi)$, $g'(y) = \arg(G(y))/(2\pi)$, $h'(z) = \arg(H(z))/(2\pi)$. A direct calculation gives that $\|f'(x) + g'(y) + h'(z)\|_{\R/\Z} \ls \tau^{49M}$, i.e. $f', g', h'$ nearly form an embedding. Applying Dirichlet approximation again gives
\[ \left|f'(x) - \frac{\a(x)}{q} \right| \le \frac{1}{qN^{1/(3M)}}, \enspace \left|g'(y) - \frac{\b(y)}{q} \right| \le \frac{1}{qN^{1/(3M)}}, \enspace \left|h'(z) - \frac{\c(z)}{q} \right| \le \frac{1}{qN^{1/(3M)}}. \]
This gives that
\begin{align*}
    \left\|\frac{1}{q}(\a(x) + \b(y) + \c(z))\right\|_{\R/\Z} \le \|f'(x) + g'(y) + h'(z)\|_{\R/\Z} + \frac{3}{qN^{1/(3M)}} < 1/q.
\end{align*}
Thus, $\a(x) + \b(y) + \c(z) \equiv 0 \pmod{q}$, i.e. forms an embedding. Thus, for any $x, x'$ with $\ama(x) = \ama(x')$ we have $\a(x) = \a(x')$ by the definition of the master embedding. Thus, for any $x, x'$ with $\ama(x) = \ama(x')$ we have that $|f'(x) - f'(x')| \le (C\tau)^{-3}$. Combining this with $|1 - |F(x)|/|F(x')|| \ls (C\tau)^{-3}$ for all $x, x'$ gives $|F(x) - F(x')| \ls \tau^{-3}$ for any $\ama(x) = \ama(x')$, contradicting the hypothesis.
\end{proof}

\subsection{Applying the Modest Noise Operator}
\label{subsec:noise2}
In \cref{sec:induction2} we will show the following bound on functions with high modest degree.
\begin{restatable}{theorem}{highmodest}
\label{thm:highmodest}
If $\mu$ satisfies the master relaxed base case (\cref{def:masterrelaxed}), then for any $G: \B^n \to \bbC$, $H: \C^n \to \bbC$, and $F: \A^n \to \bbC$ that with all terms having modest degree at least $d$:
\[ \left|\E_{(x,y,z) \sim \mu^{\otimes n}}\left[F(x) G(y) H(z) \right] \right| \le 2^{-c\cdot d(d/n)^C} \|F\|_2 \|G\|_2 \|H\|_2.\]
\end{restatable}

We will apply this bound to reduce to studying functions that are constant on the master embedding. The idea is to apply the modest noise operator, apply \cref{thm:highmodest} to remove the high-degree terms, and then take a random restriction onto a distribution that distributes mass equally over inputs with the same master embedding.

\begin{lemma}
\label{lemma:reducef}
Let $\mu$ be a saturated distribution on $\A \times \B \times \C$, and let $\nu$ be a distribution on $\cA \times \B \times \C$ sampled as follows: sample $(y, z) \sim \mu_{y,z}$, and return $(-\ama(y)-\ama(z), y, z)$.

For any constant $c' > 0$, for all sufficiently small $c$ (in terms of $c'$), and $\val_{t,\mu}(f, g, h) \ge \eps$ for $\eps \ge 2^{-cn}$, there are strategies $\wt{f}, \wt{g}, \wt{h}$ on $n/8$ coordinates with $\val_{\nu, t}(\wt{f}, \wt{g}, \wt{h}) \gs \eps^{O_M(1)}/2^{c'n}.$
\end{lemma}
\begin{proof}
Perform path tricks and merging as in \cref{lemma:masterrelaxed} to transform $\mu$ into $\mu^+$ which satisfies the master relaxed base case. For simplicity, relabel $f, g, h$ as the resulting strategies. Note that this reduces $n$ by a constant, due to random restrictions, and $\val_{\mu^+,t}(f, g, h) \ge \eps'$ for $\eps' = \eps^{O_M(1)}$.

Let $d \in \Z_{>0}$ be chosen later, $\hat{\eps} = 1/4$, and $\{c_i\}, \{\rho_i\}$ as in \cref{lemma:polyapprox} for $\eps := \hat{\eps}$. Let $\TT = \sum_{j=0}^d c_j \TT^{\modest}_{\rho_j}$. Let $F_S, G_S, H_S$ be as defined in \cref{lemma:identity}.
Let $a, b, c$ and $N$ satisfy $a(x) + b(y) + c(z) = N \cdot t(x, y, z) \pmod{p^k N}$, and define
\[ A(x) := \exp\left(\frac{2\pi i \cdot a(x)}{p^k N} \right), \enspace B(y) := \exp\left(\frac{2\pi i \cdot b(y)}{p^k N}\right), \text{ and } \enspace C(z) := \exp\left(\frac{2\pi i \cdot c(z)}{p^k N}\right). \] For $S \subseteq \Z_m^n$, define $A_S(x) = \prod_{i\in[n]} A(x)^{S_i}$, so that $T_S(x,y,z) = A_S(x)B_S(y)C_S(z)$ for all $(x, y, z) \in \supp(\mu)$.
By \cref{lemma:identity},
\begin{align*} 
\val_{\mu^+,t}(f,g,h) &= \E_{(x,y,z)\sim(\mu^+)^{\otimes n}}\left[F_S(x) G_S(y) H_S(z) T_S(x,y,z) \right] \\ &= \E_{(x,y,z)\sim(\mu^+)^{\otimes n}}\left[(I-\TT)(A_SF_S)(x) (B_SG_S)(y) (C_SH_S)(z) \right] \\ &+ \E_{(x,y,z)\sim(\mu^+)^{\otimes n}}\left[\TT(A_SF_S)(x) (B_SG_S)(y) (C_SH_S)(z) \right]. \end{align*}
We will bound the first term using \cref{thm:highmodest}. Indeed,
\begin{align*}
    &\left|\E_{(x,y,z)\sim(\mu^+)^{\otimes n}}\left[(I-\TT)(A_SF_S)(x) (B_SG_S)(y) (C_SH_S)(z) \right]\right| \\ = ~& \left|\sum_{k=0}^n (1-\sum_{j=0}^d c_j \rho_j^k) \E_{(x,y,z)\sim(\mu^+)^{\otimes n}}[(A_SF_S)^{\modest=k}(x) (B_SG_S)(y) (C_SH_S)(z)] \right| \\
    \le ~&\sum_{k=d+1}^n 2^{-ck(k/n)^C} \le n \cdot 2^{-cd(d/n)^C}, 
\end{align*}
where we have used \cref{lemma:polyapprox} and \cref{lemma:highmodest}. For the choice $d = c'n$ for $c' \asymp c^{1/(C+1)}$, we know that $2^{-\Omega(d(d/n)^C)} \le \eps'/2$, because $\eps' \ge 2^{-cn}$. Thus \cref{lemma:identity} and triangle inequality give
\begin{align*} &\sum_{j=0}^d |c_j| \left|\E_{S\sim\Z_m^n}\E_{(x,y,z) \sim (\mu^+)^{\otimes n}}[\TT^{\modest}_{\rho_j}(A_SF_S)(x) (B_SG_S)(y) (C_SH_S)(z)]\right| \\ \ge ~&\val_{\mu^+,t}(f, g, h) - \eps'/2 \ge \eps'/2.
\end{align*}
Because $\sum_{j=1}^k |c_j| \le 2^{O(d)}$ by \cref{lemma:polyapprox}, there is some $j$ with 
\[ \left|\E_{S\sim\Z_m^n}\E_{(x,y,z) \sim (\mu^+)^{\otimes n}}[\TT^{\modest}_{\rho_j}(A_SF_S)(x) (B_SG_S)(y) (C_SH_S)(z)]\right| \ge \eps'/2^{O(d)} \gs \eps^{O_M(1)}/2^{O(c'n)}. \]
Define the distribution $\mu'$ as follows: take a sample $(x, y, z) \sim \mu^+$, and then resample $x' \sim N^{\modest}_{1-\rho_j}(x)$. Then the above expression equals
\[ \left|\E_{S\sim\Z_m^n}\E_{(x,y,z) \sim (\mu')^{\otimes n}}[(A_SF_S)(x) (B_SG_S)(y) (C_SH_S)(z)]\right| = \val_{t,\mu'}(f, g, h), \]
by \cref{lemma:identity} again. Thus $\val_{t,\mu'}(f, g, h) \gs \eps^{O_M(1)}/2^{O(c'n)}$. Formally, we need to define the tensor $t$ on the support of $\mu'$. However, by \cref{lemma:ass} we know that for any $x$, $\ama$ is constant on the support of $N^{\modest}_{1-\rho_j}(x)$, so $t$ just inherits its value.

To conclude the proof, we will take a random restriction of $\mu'$. Note that all atoms have in $\mu'$ value at least $1-\rho_j \ge \hat{\eps}$ for all $j = 0, 1, \dots, d$, by \cref{lemma:polyapprox}. By construction, $\mu^+$ contained all points of the form $(\rep((x, \dots, x)), (y, \dots, y), z)$ for $(x, y, z) \in \supp(\mu)$. Thus we can take a random restriction onto a distribution $\nu$ supported on $(x', y, z)$ for all $(y, z) \in \supp(\mu)$ and $\ama(x') = \ama(\rep(x)) = -\ama(y)-\ama(z)$. The random restriction reduces $n$ to at least $n' = \hat{\eps} n/2 = n/8$ coordinates, and only affects the winning probability by an exponentially small amount (see \cref{lemma:rr}).
\end{proof}

Apply \cref{lemma:reducef} to reduce to a distribution $\mu$ on $\cA \times \B \times \C$. Because $(y, z)$ is already a product distribution in the master relaxed base case, we can take another random restriction so make $\bma(y), \cma(z)$ uniform over $\cA$. Thus, we have reduced to the distribution on $\cA^3$ which is uniform over all $(x, y, z) \in \cA^3$ with $x+y+z=0$. \cref{thm:additive}, whose proof is based on additive combinatorics (as in \cite{BKM22}), shows that in this case, $\eps \le 2^{-cn}$. This completes the proof of \cref{thm:main}, after we show \cref{thm:highmodest} and \cref{thm:additive}.

\section{Main Induction: Complex Embeddings}
\label{sec:induction2}
This section is dedicated to proving \cref{thm:highmodest}. The proof is nearly identical to that of \cref{thm:highdegree}.
\highmodest*
Define $\cQ^{(n,d)}$ to be the set of functions $F: \A^n \to \bbC$ all of whose terms have modest degree at least $d$. Define
\[ \gamma_{n,d} := \sup_{\substack{\|F\|_2 \le 1, F \in \cQ^{(n,d)} \\ \|G\|_2 \le 1, \|H\|_2 \le 1}} \left|\E_{(x,y,z)\sim\mu^{\otimes n}}[F(x)G(y)H(z)]\right|. \]
Then \cref{thm:highmodest} follows from the following lemma.
\begin{lemma}
\label{lemma:highmodest}
If $\mu$ satisfies the master relaxed base case, then $\gamma_{n,d} \le (1-c(d/n)^{O_M(1)})\gamma_{n-1,d-1}$.
\end{lemma}
\begin{proof}[Proof of \cref{thm:highmodest}]
Apply \cref{lemma:highmodest} a total of $d/2$ times, giving \[ \gamma_{n,d} \le (1 - c(d/n)^{O_M(1)})^{d/2} \le \exp(-c \cdot d(d/n)^{O_M(1)}). \qedhere\]
\end{proof}
The remainder of the section is devoted to establishing \cref{lemma:highmodest}. Because $F \in \cQ^{(n,d)}$ we know that $I_{\modest}[F] \ge d$, so $\tau := I_{j,\modest}[F] \ge d/n$ for some $j$.  Without loss of generality, let $j = n$, and perform an SVD (\cref{lemma:svd}) to write $F(x) = \sum_{r=1}^M \aa_r F_r(x_I) F_r'(x_n)$.

Analogous to \cref{lemma:variance}, the sum of modest influences of the $F_r'$ are large.
\begin{lemma}
\label{lemma:modestvariance}
We have: $\sum_{r\in[M]} \aa_r^2 I_{\modest}[F_r'] = I_{n,\modest}[F] \ge \tau$.
\end{lemma}
\begin{proof}
Indeed, if $(x_n,y_n)$ are sampled by the distribution used to define $I_{\modest}$, we know
\begin{align*}
I_{n,\modest}[F] &= \E_{x_I,(x_n,y_n)}[|F(x_I,x_n) - F(x_I,y_n)|^2] \\ &= \E_{x_I,(x_n,y_n)}\left[\left|\sum_{r\in[M]} \aa_r F_r(x_I)(F_r'(x_n) - F_r'(y_n)) \right|^2\right] \\
&= \E_{x_I,(x_n,y_n)}\left[\sum_{r_1,r_2\in[M]} \aa_{r_1}\aa_{r_2}F_{r_1}(x_I)\bar{F_{r_2}(x_I)} (F_{r_1}'(x_n) - F_{r_1}'(y_n))\bar{(F_{r_2}'(x_n) - F_{r_2}'(y_n))}\right] \\
&= \E_{(x_n,y_n)}\left[\sum_{r\in[M]} \aa_r^2 |F_r'(x_n) - F_r'(y_n)|^2 \right] = \sum_{r\in[M]} \aa_r^2 I_{\modest}[F_r']
\end{align*}
\end{proof}
Given the SVDs
\[ F(x) = \sum_{r\in[M]} \aa_r F_r(x_I) F_r'(x_J), \enspace G(y) = \sum_{s\in[M]} \bb_s G_s(x_I) G_s'(x_J), \enspace H(z) = \sum_{t\in[M]} \cc_t H_t(z_I) H_t'(z_J), \]
define the quantities
\begin{align*} \hat{F_r}(s, t) &:= \E_{(x,y,z)\sim\mu^{\otimes I}}[F_r(x) G_s(y) H_t(z)] \\
\hat{F_r'}(s, t) &:= \E_{(x,y,z)\sim\mu}[F_r'(x)G_s'(y)H_t'(z)].
\end{align*}
Then by definition we have that
\[ \E_{(x,y,z)\sim\mu^{\otimes n}}[F(x)G(y)H(z)] = \sum_{r,s,t\in[M]} \aa_r\bb_s\cc_t \hat{F_r}(s, t)\hat{F_r'}(s, t). \]
\subsection{Finding a Singular Value Gap}
\label{subsection:complexgap}
Let $\delta < \Delta$ satisfy that none of the $\aa_r, \bb_s, \cc_t$ lie in $[\delta,\Delta]$, and $\delta \le \left(\frac{\Delta}{CM\tau}\right)^{1000M}$.
We can find such $\delta, \Delta$ with $\delta \ge \tau^{O_M(1)}$. Let $R = \{r : \aa_r \ge \Delta\}, S = \{s : \bb_s \ge \Delta\}, T = \{t : \cc_t \ge \Delta\}$. Note that because $|\hat{F_r}(s,t)| \le \gamma_{n-1,d-1}, |\hat{F_r'}(s,t)| \le 1$ we know
\begin{align} \left|\sum_{r,s,t\in[M]} \aa_r\bb_s\cc_t\hat{F_r}(s,t)\hat{F_r'}(s,t)\right| \le \left|\sum_{r\in R,s\in S,t\in T} \aa_r\bb_s\cc_t\hat{F_r}(s,t)\hat{F_r'}(s,t)\right| + \delta M^3 \gamma_{n-1,d-1}. \label{eq:remove2} \end{align}

\subsection{Cauchy-Schwarz and Simple Observations}
\label{subsec:firstcs2}
For the remainder of the section, we assume that \cref{lemma:highmodest} is false for the choice of $F, G, H$. By the Cauchy-Schwarz inequality we get:
\begin{align} \left|\sum_{r\in R,s\in S,t\in T} \aa_r\bb_s\cc_t\hat{F_r}(s,t)\hat{F_r'}(s,t)\right| \le \sqrt{\sum_{s\in S,t\in T} \bb_s^2 \cc_t^2 \sum_{r\in R} |\hat{F_r}(s,t)|^2}\sqrt{\sum_{r\in R} \aa_r^2 \sum_{s\in S,t\in T} |\hat{F_r'}(s,t)|^2 }. \label{eq:cs2} \end{align}
The following simple observations are analogous to their counterparts \cref{lemma:bound1,lemma:bound21,lemma:bound2,cor:bound1,cor:bound2} in \cref{sec:induction}.
\begin{lemma}
\label{lemma:cbound1}
For all $s, t$ we have: $\sum_{r\in[M]} |\hat{F_r}(s,t)|^2 \le \gamma_{n-1,d-1}^2$.
\end{lemma}
\begin{lemma}
\label{lemma:cbound21}
For all $s, t$ we have: $\sum_{r\in[M]} |\hat{F_r'}(s,t)|^2 \le 1$.
\end{lemma}
\begin{lemma}
\label{lemma:cbound2}
For all $r$ we have: $\sum_{s,t\in[M]} |\hat{F_r'}(s,t)|^2 \le 1$.
\end{lemma}
\begin{lemma}
\label{cor:cbound1}
For all $s\in S,t\in T$ we have: $\sum_{r\in R} |\hat{F_r}(s,t)|^2 \ge (1 - \frac{4\delta M^3}{\Delta^4})\gamma_{n-1,d-1}^2$.
\end{lemma}
\begin{lemma}
\label{cor:cbound2}
For all $r\in R$ we have: $\sum_{s\in S,t\in T} |\hat{F_r'}(s,t)|^2 \ge 1 - \frac{4\delta M^3}{\Delta^2}$.
\end{lemma}
The case $|R| < |S||T|$ is handled identically to \cref{subsec:r<st} (imagine that the tensor $T \equiv 1$). So in the remainder of the section, we focus on the case $|R| \ge |S||T|$.

\subsection{Handling the Case \texorpdfstring{$|R| \ge |S||T|$}{r>st2}}
\label{subsec:r>st2}
In this case we use the master relaxed base case (\cref{def:relaxed}). We use the analogue of \cref{lemma:fr'large}, whose proof is identical.
\begin{lemma}
\label{lemma:fr'large2}
If $|R| \ge |S||T|$ then for all $s \in S, t \in T$ we have $\sum_{r \in R} |\hat{F_r'}(s,t)|^2 \ge 1 - \frac{4\delta M^5}{\Delta^2}$.
\end{lemma}
For $s \in S, t \in T$, define the functions
\[ \wt{F_{st}} = \frac{\sum_{r \in R} \bar{\hat{F_r'}(s,t)}F_r'}{\sqrt{\sum_{r\in R} |\hat{F_r'}(s,t)|^2}}. \]
Note that $\|\wt{F_{st}}\|_2 = 1$ and
\begin{align} \left(\E_{(x,y,z)\sim\mu}\left[\wt{F_{st}}(x) G_s'(y) H_t'(z) \right]\right)^2 = \sum_{r \in R} |\hat{F_r'}(s,t)|^2 \ge 1 - \frac{4\delta M^5}{\Delta^2}, \label{eq:cfr'2} \end{align}
by \cref{lemma:fr'large2}. Now we leverage the master relaxed base case property of $\mu$ to get that \[ I_{\modest}[\wt{F_{st}}] \le \left(\frac{4\delta M^5}{c\Delta^2}\right)^{\frac{1}{100M}} := \delta', \] for all $s \in S$ and $t \in T$. The remainder of the proof converts this into an effective influence bound on the $F_r'$ functions themselves to contradict \cref{lemma:modestvariance}.

In the following, treat all functions to be on the space $L_2(\B \times \C, \mu_{y,z})$, where we recall that $(y,z)$ is a product distribution and determines $x$ uniquely in the relaxed base case. Note that for all $s \in S, t \in T$ that
\begin{align} \|\bar{\wt{F_{st}}} - G_s' H_t'\|_2 = \|\wt{F_{st}}\|_2^2 + \|G_s' H_t'\|_2^2 - 2 \E_{(x,y,z)}\left[\wt{F_{st}}(x) G_s'(y) H_t'(z) \right] \le \frac{8\delta M^5}{\Delta^2}, \label{eq:kobe2} \end{align} where we have used \eqref{eq:cfr'2} and that the expectation is real. For all $r \in R$ we have
\[ \Big\|\bar{F_r'} - \sum_{s\in S, t\in T} \bar{\hat{F_r'}(s,t)} \bar{\wt{F_{st}}} \Big\|_2 \le \Big\|\bar{F_r'} - \sum_{s \in S, t \in T} \bar{\hat{F_r'}(s,t)} G_s' H_t' \Big\|_2 + \sum_{s \in S, t \in T} \hat{F_r'}(s,t)\left\|G_s' H_t' - \bar{\wt{F_{st}}} \right\|_2. \]
By \eqref{eq:kobe2} the second term can be bounded by
\[ \sum_{s \in S, t \in T} \hat{F_r'}(s,t)\left\|G_s' H_t' - \bar{\wt{F_{st}}} \right\|_2 \le M^2 \cdot \frac{8\delta M^5}{\Delta^2} = \frac{8\delta M^7}{\Delta^2}. \]
The first term can be simplified to get
\[ \Big\|\bar{F_r'} - \sum_{s \in S, t \in T} \bar{\hat{F_r'}(s,t)} G_s' H_t' \Big\|_2 = \Big(\|F_r'\|_2^2 - \sum_{s \in S, t \in T} |\hat{F_r'}(s,t)|^2 \Big)^{1/2} \le \sqrt{\frac{4\delta M^3}{\Delta^2}},\]
where we have used that the functions $G_s'H_t'$ are orthonormal (because $\mu_{y,z}$ is a product distribution), and \cref{cor:cbound2}. By the triangle inequality again, we get
\begin{align*} I_{\modest}[F_r'] &\le I_{\modest}\Big[\bar{F_r'} - \sum_{s\in S, t\in T} \bar{\hat{F_r'}(s,t)} \bar{\wt{F_{st}}}\Big] + \sum_{s \in S, t \in T} |\hat{F_r'}(s,t)| I_{\modest}[\wt{F_{st}}]
\\ &\le \Big\|\bar{F_r'} - \sum_{s\in S, t\in T} \bar{\hat{F_r'}(s,t)} \bar{\wt{F_{st}}}\Big\|_2 + \sum_{s \in S, t \in T} |\hat{F_r'}(s,t)| I_{\modest}[\wt{F_{st}}] \\ 
&\le \sqrt{\frac{4\delta M^3}{\Delta^2}} + \frac{8\delta M^7}{\Delta^2} + M^2 \delta' \le 2M^2 \delta',
\end{align*}
by our choice of $\delta = \left(\frac{\Delta}{CM\tau}\right)^{1000M}.$ This gives that
\[ \sum_{r \in [M]} \aa_r^2 I_{\modest}[F_r'] \le M \delta^2 + \sum_{r \in R} \aa_r^2 I_{\modest}[F_r'] \le M\delta^2 + 2M^2 \delta' < \tau, \]
by the choice of $\delta$ again, because $\delta' = \left(\frac{4\delta M^5}{c\Delta^2}\right)^{\frac{1}{100M}}$. This contradicts \cref{lemma:modestvariance}.

\section{Additive Combinatorics}
\label{sec:additive}
This section is dedicated to proving \cref{thm:additive}.
\additive*

We show \cref{thm:additive} in the remainder of this section. By minimality, we may assume that $N$ is a power of $p$. Let $N = p^j$. 
Assume that \[ \Pr_{\substack{(x,y,z) \in (\cA^n)^3 \\ x+y+z=0}}[f(x)_i + g(y)_i + h(z)_i = t(x_i,y_i,z_i) \pmod{p^k}, i = 1, 2, \dots, n] \ge \eps, \]
where $\eps > 2^{-cn}$, for contradiction. Then writing $\wt{f}(x) = N \cdot f(x) - a(x)$, $\wt{g}(y) = N \cdot g(y) - b(x)$, and $\wt{h}(z) = N \cdot h(z) - c(z)$, where $\wt{f}, \wt{g}, \wt{h}: \cA^n \to \Z_{p^kN}^n$, we get that
\[ \Pr_{\substack{(x,y,z) \in (\cA^n)^3 \\ x+y+z=0}}[\wt{f}(x) + \wt{g}(y) + \wt{h}(z) = 0 \pmod{p^kN}, i = 1, 2, \dots, n] \ge \eps. \]
We can follow the approach \cite[Section 2]{BKM22} to conclude the following.
\begin{lemma}
\label{lemma:freiman}
    There is a subset $\cA' \subseteq \cA^n$ with $|\cA'| \ge \Omega(\eps^{O(p^k N)}) |\cA|^n$ such that $\wt{f}$ is an order-$p^kN$ Freiman homomorphism on $\cA'$.
\end{lemma}
At this point, write $\cA = \prod_{i=1}^L \Z_{p_i^{k_i}}$. We have the map $a: \cA \to \Z_{p^kN}$. Consider its restriction modulo $p$, which we write as $\bar{a}: \cA \to \Z_p$ defined as $\bar{a}(x) = a(x) \pmod{p}$. We show that $\bar{a}$ must be the shift of a homomorphism.
\begin{lemma}
\label{lemma:ahom}
For sufficiently small constants $c$, if $\eps \ge 2^{-cn}$ then $\bar{a}: \cA \to \Z_p$ must take the form $\bar{a}(x) = c + \sum_{i=1}^L c_i x_i \pmod{p}$, where $x = (x_1, \dots, x_L) \in \cA$.
\end{lemma}
\begin{proof}
    Because $\wt{f}$ is an order $4$ Freiman homomorphism we know for all $u,v,w,x \in (\cA')^4$ with $u+v = w+x$ we have $\wt{f}(u) + \wt{f}(v) = \wt{f}(w) + \wt{f}(x) \pmod{p^kN}$. Recall that $\wt{f}(x) = N \cdot f(x) - a(x) = -\bar{a}(x) \pmod{p}$.
    So if $u+v=w+x$ then $\bar{a}(u)+\bar{a}(v) = \bar{a}(w) + \bar{a}(x) \pmod{p}$. The number of such $4$-tuples is at least $|\cA'|^4/|\cA|^n \ge \Omega(\eps^{O(p^kN)}) |\cA|^{3n}$, so $\bar{a}$ satisfies:
    \[ \Pr_{(u,v,w,x) \in (\cA^n)^4: u+v=w+x}[\bar{a}(u_i) + \bar{a}(v_i) = \bar{a}(w_i) + \bar{a}(x_i) \pmod{p} \enspace \forall \enspace i \in [n]] \ge \Omega(\eps^{O(p^kN)}). \]
    If $\bar{a}(u) + \bar{a}(v) \neq \bar{a}(w) + \bar{a}(x)$ for some $(u,v,w,x) \in \cA^4$ with $u+v=w+x$ then the above probability is at most $(1-1/|\cA|^3)^n \ll \eps^{O(p^kN)}$.
    So $\bar{a}(u) + \bar{a}(v) = \bar{a}(w) + \bar{a}(x)$ for all $u+v=w+x$. This implies the conclusion.
\end{proof}
Now we show that many of the coefficients $c_i = 0$ in fact.
\begin{lemma}
    \label{lemma:czero}
    Under the hypotheses of \cref{lemma:ahom} we have that $c_i = 0$ if $p_i \neq p$, or $p_i = p$ and $k_i < k+j$.
\end{lemma}
\begin{proof}
If $c_i \neq 0$ for some $p_i \neq p$, then $\bar{a}$ is not the shift of a homomorphism from $\cA \to \Z_p$.

Now consider $i$ with $p_i = p$ and $k_i < k+j$. Split $\cA^n$ into $|\cA|^n / p^n$ groups as follows.
We first describe how to split $(\Z_{p^{k_i}})^n$ into $(p^{k_i})^n / p^n$ groups. This is done by splitting $\Z_{p^{k_i}}$ into $p^{k_i-1}$ groups of size $p$:
\[ \{0, 1, \dots. p-1\} \cup \{p, p+1, \dots, 2p-1\}, \dots, \{p^{k_i}-p, \dots, p^{k_i}-1\}, \]
and then taking the direct product $n$ times.
The point is that for any $x, y \in (\Z_{p^{k_i}})^n$ in the same group, $x \neq y \pmod{p}$, as $n$-dimensional vectors.
Now we split $\cA^n$ as follows: $(x_1, \dots, x_L)$ and $(y_1, \dots, y_L)$ are in the same group if and only if $x_{i'} = y_{i'}$ for any $i' \neq i$, and $(x_i, y_i)$ are in the same group of $(\Z_{p^{k_i}})^n$ as described.

Because $|\cA'| > |\cA|^n / p^n$, the Pigeonhole Principle tells us that we can find $x, y \in \cA'$ in the same group. By construction, $p^{k_i}x = p^{k_i}y$. Because $\wt{f}$ was an order $p^kN > p^{k_i}$ Freiman homomorphism, we deduce that $p^{k_i}\wt{f}(x) = p^{k_i}\wt{f}(y) \pmod{p^{k+j}}$, so $\wt{f}(x) = \wt{f}(y) \pmod{p}$ as $k_i < k+j$. Thus, $\bar{a}(x) = \bar{a}(y)$. Because $x_{i'} = y_{i'}$ for all $i' \neq i$ by construction, using \cref{lemma:ahom} gives us $c_i(x_i-y_i) = 0 \pmod{p}$. By construction $x_i \neq y_i \pmod{p}$ as vectors, so $c_i = 0$.
\end{proof}
By symmetry, we can apply both \cref{lemma:ahom,lemma:czero} to $b, c$ to write:
\[ \bar{b}(y) = d + \sum_{i=1}^L d_i y_i \pmod{p} \enspace \text{ and } \enspace \bar{c}(z) = e + \sum_{i=1}^L e_i z_i \pmod{p}. \]
We will use this to contradict minimality of $N$. To start, we first show a relationship between the coefficients $c_i, d_i, e_i$.
\begin{lemma}[Matching coefficients]
    \label{lemma:coef}
    If $\bar{a}(x) = c + \sum_{i=1}^L c_i x_i$, $\bar{b}(y) = d + \sum_{i=1}^L d_i y_i$, and $\bar{c}(z) = e + \sum_{i=1}^L e_i z_i$ then $c+d+e=0\pmod{p}$ and $c_i = d_i = e_i$ (as elements of $\Z_p$).
\end{lemma}
\begin{proof}
Recall that $\bar{a}(x) + \bar{b}(y) + \bar{c}(z) = N \cdot t(x,y,z) = 0 \pmod{p}$ for all $x+y+z = 0$ in $\cA$. Thus $c+d+e=0\pmod{p}$ and $c_i=d_i=e_i$.
\end{proof}
Finally we reach a contradiction by showing that $N$ was not minimal. 
\begin{lemma}
\label{lemma:reduceN}
Assume $c+d+e=0$ over $\Z$.
Define $\wt{a}(x) = a(x) - c - \sum_{i=1}^L c_i x_i$, $\wt{b}(y) = b(y) - d - \sum_{i=1}^L c_i y_i$, $\wt{c}(z) = c(z) - e - \sum_{i=1}^L c_i z_i$.
Then $\wt{a}(x) + \wt{b}(y) + \wt{c}(z) = a(x) + b(y) + c(z) \pmod{p^kN}$ and $\wt{a}(x), \wt{b}(y), \wt{c}(z) = 0 \pmod{p}$ for any $x, y, z \in \cA$.
\end{lemma}
\begin{proof}
    The final claim from \cref{lemma:ahom}.
    For all $(x,y,z) \in \cA^3$ with $x+y+z = 0$ we know that
    \[(a(x) + b(y) + c(z)) - (\wt{a}(x) + \wt{b}(y) + \wt{c}(z)) = (c+d+e) + \sum_{i=1}^L c_i(x_i+y_i+z_i). \]
    If $p_i = p$ and $k_i < k+j$ then $c_i = 0$. Otherwise, $x_i + y_i + z_i = 0 \pmod{p^kN}$, as $N = p^j$. In both cases, $\sum_{i=1}^L c_i(x_i+y_i+z_i) = 0 \pmod{p^kN}$ as desired.
\end{proof}
This contradicts the minimality of $N$, because we can shift $c, d, e$ by multiples of $p$ to ensure that $c+d+e=0$ (recall that $c+d+e=0\pmod{p}$ by \cref{lemma:coef}), and then divide everything by $p$ to reduce $N$ to $N/p$.

{\small
\bibliographystyle{alpha}
\bibliography{refs}}

\appendix

\section{Proof of \texorpdfstring{\cref{lemma:polyapprox}}{polyapprox}}
\label{app:poly}

\begin{proof}
Fix $\rho_0, \dots, \rho_d \in [0, 1-\eps]$ to be chosen later, and let $c_j = \frac{\prod_{i\neq j}(1-\rho_i)}{\prod_{i\neq j}(\rho_j - \rho_i)}$ for $j = 0, \dots, d$. We will show that \cref{item:degree<d,item:degree>d} of \cref{lemma:polyapprox} are satisfied for any choice of $\rho_j$.

For \cref{item:degree<d} (where $k \le d$) consider the polynomial $p(x) = x^k$. By the Lagrange interpolation formula, we know that
\[ p(x) = \sum_{j=0}^d \frac{\prod_{i\neq j}(x-\rho_i)}{\prod_{i\neq j}(\rho_j-\rho_i)} p(\rho_j) = \sum_{j=0}^d \frac{\prod_{i\neq j}(x-\rho_i)}{\prod_{i\neq j}(\rho_j-\rho_i)} \rho_j^k. \] Taking $x = 1$ gives us
$1 = p(1) = \sum_{j=0}^d c_j \rho_j^k$, as desired.

Define $A_k = \sum_{j=0}^d c_j \rho_j^k$. We will show that $A_k \ge A_{k+1}$ for all $k \ge 0$. Because $\lim_{k\to\infty} A_k = 0$ evidently, and $A_d = 1$ by the above, $0 \le A_k \le 1$ for all $k \ge d+1$, as desired. Note that
\[ A_k - A_{k+1} = \sum_{j=0}^d \frac{\prod_{i\neq j}(1-\rho_i)}{\prod_{i\neq j}(\rho_j - \rho_i)} (\rho_j^k - \rho_j^{k+1}) = \prod_{j=0}^d (1-\rho_j) \cdot \sum_{j=0}^d \frac{\rho_j^k}{\prod_{i\neq j} (\rho_j - \rho_i)}. \]
The positivity of this expression follows from $\prod_{j=0}^d (1-\rho_j) > 0$, as $\rho_0, \dots, \rho_d < 1$, and the identity
\begin{align}
\sum_{j=0}^d \frac{\rho_j^k}{\prod_{i\neq j} (\rho_j - \rho_i)} = \sum_{\substack{a_0 + \dots + a_d = k-d \\ a_0, \dots, a_d \in \Z_{\ge0}}} \rho_0^{a_0} \rho_1^{a_1} \dots \rho_d^{a_d} \ge 0. \label{eq:schur}
\end{align}
We establish this identity below, in \cref{lemma:schur}.

We proceed to \cref{item:totalc}. Let
\begin{equation}
    T_d(x) := \frac12\left((x+\sqrt{x^2-1})^d + (x-\sqrt{x^2-1})^d \right) \label{eq:cheby}
\end{equation}
denote the $d$-th Chebyshev polynomial, and let $\tilde{T}(x) = T_d(\frac{2}{1-\eps}x - 1)$. Let $\rho_j = \frac{1-\eps}{2}(\cos(j\pi/d) + 1)$. Note that $0 \le \rho_j \le 1-\eps$, and $\tilde{T}(\rho_j) = T_d(\cos(j\pi/d)) = \cos(j\pi) = (-1)^j$. Let $a_0, \dots, a_d$ be the coefficients of $\tilde{T}$, i.e., $\tilde{T}(x) = \sum_{k=0}^d a_k x^k$. We have that
\begin{align*}
\tilde{T}(1) =
\sum_{k=0}^d a_k = \sum_{k=0}^d a_k \sum_{j=0}^d c_j \rho_j^k = \sum_{j=0}^d c_j \sum_{k=0}^d a_k \rho_j^k = \sum_{j=0}^d c_j \tilde{T}(\rho_j) = \sum_{j=0}^d c_j (-1)^j = \sum_{j=0}^d |c_j|,
\end{align*}
where the final equality uses the definition $c_j = \frac{\prod_{i\neq j}(1-\rho_i)}{\prod_{i\neq j}(\rho_j - \rho_i)}$ to deduce that $c_j(-1)^j > 0$ for all $j$ (recall that $\rho_0 > \rho_1 > \dots > \rho_d$). To conclude, note that \[ \tilde{T}(1) = T\left(\frac{1+\eps}{1-\eps}\right) \le T(1+4\eps) \le \exp(O(d\sqrt{\eps})), \]
where the final inequality comes from taking $x = 1+4\eps$ in \eqref{eq:cheby}.
\end{proof}
\begin{lemma}
\label{lemma:schur}
\eqref{eq:schur} holds for any distinct real numbers $\rho_0, \dots, \rho_d$, and integer $k \ge d$.
\end{lemma}
\begin{proof}
Using the Vandermonde determinant, one can compute that
\begin{align*}
\sum_{j=0}^d \frac{\rho_j^k}{\prod_{i\neq j} (\rho_j - \rho_i)} = \det \begin{bmatrix}
    \rho_0^k & \rho_1^k & \dots & \rho_d^k \\
    \rho_0^{d-1} & \rho_1^{d-1} & \dots & \rho_d^{d-1} \\
    \vdots & \vdots & \ddots & \vdots \\
    1 & 1 & \dots & 1
\end{bmatrix} \Bigg{/}
\det \begin{bmatrix}
    \rho_0^d & \rho_1^d & \dots & \rho_d^d \\
    \rho_0^{d-1} & \rho_1^{d-1} & \dots & \rho_d^{d-1} \\
    \vdots & \vdots & \ddots & \vdots \\
    1 & 1 & \dots & 1
\end{bmatrix}.
\end{align*}
This is the Schur polynomial corresponding to the partition $\lambda := (0, 0, \dots, k-d)$. The Giambelli formula says that this the Schur polynomial can be alterantively expressed as
\[ \sum_T \rho_0^{t_0} \dots \rho_d^{t_d}, \]
where $T$ ranges over all semistandard Young Tableaux for partition $\lambda$, and $t_i$ is the number of occurrences of $i$ in $T$. Evidently, this is equivalent to the RHS of \eqref{eq:schur}, for $\lambda$ with a single piece of size $k-d$, and labels between $0, \dots, d$.
\end{proof}

\end{document}